\documentclass{article}%
\usepackage{amssymb}
\usepackage{amsmath}
\usepackage{amsfonts}
\usepackage{graphicx}%
\providecommand{\U}[1]{\protect\rule{.1in}{.1in}}
\providecommand{\U}[1]{\protect\rule{.1in}{.1in}}
\newtheorem{theorem}{Theorem}[section]

\newtheorem{conjecture}[theorem]{Conjecture}

\newtheorem{definition}[theorem]{Definition}

\newtheorem{lemma}[theorem]{Lemma}

\newtheorem{proposition}[theorem]{Proposition}
\newtheorem{remark}[theorem]{Remark}

\newenvironment{proof}[1][Proof]{\noindent\textbf{#1.} }{\ \rule{0.5em}{0.5em}}
\begin{document}

\title{The planar 3-body problem $\mathrm{II}$: reduction to pure shape and spherical
geometry (2nd version)}
\author{Wu-Yi Hsiang\\Department of Mathematics\\University of California, Berkeley
\and Eldar Straume\\Department of Mathematical Sciences\\Norwegian University of Science and \\Technology, Trondheim, Norway}
\maketitle

\begin{abstract}
Geometric reduction of the Newtonian planar three-body problem is investigated
in the framework of equivariant Riemannian geometry, which reduces the study
of trajectories of three-body motions to the study of their moduli curves,
that is, curves which record the change of size and shape, in the moduli space
of oriented mass-triangles. The latter space is a Riemannian cone over the
shape sphere $\simeq S^{2}$, and the shape curve is the image curve on this
sphere. It is shown that the time parametrized moduli curve is in general
determined by the relative geometry of the shape curve and the shape potential
function. This also entails the reconstruction of time, namely the geometric
shape curve determines the time parametrization of the moduli curve, hence
also of the three-body motion itself, modulo a fixed rotation of the plane.

The first version of this work is an unpublished paper from 2012, and the
present version is an editorial revision of this.

\end{abstract}
\tableofcontents

\section{Introduction}

Traditionally, the 3-body problem in celestial mechanics is most often studied
in the framework of Hamiltonian mechanics, cf. e.g. \cite{SM} . However, since
1994 (cf. \cite{HS-0}) the authors have studied the kinematic geometry of
3-body systems and geometric reduction of 3-body motions, in the framework of
equivariant Riemannian geometry and inspired by Jacobi's geometrization of
Lagrange's least action principle. For a general introduction to this
reduction approach we refer to the monograph\ \cite{HS-1}.

In the general study of the 3-body problem, the Newtonian dynamical equations
are formulated at the level of the configuration space $M$, which consists of
all ordered positions of the three bodies in a Euclidean 3-space, called
\emph{m-triangles}, and the solutions $t\rightarrow\gamma(t)$ of the 3-body
motions are the trajectories of these equations. The group of rigid motions
acts naturally on m-triangles, whose orbits are the congruence classes of
m-triangles, and we shall denote by $\bar{M}$ the \emph{moduli space}
consisting of all congruence classes.\emph{ }The image $\bar{\gamma}(t)$ of
the trajectory $\gamma(t)$ in $\bar{M}$ is referred to as the \emph{moduli
curve}. This is the first step of the geometric reduction, namely the
reduction of the 3-body problem to dynamics and analysis in the quotient space
$\bar{M}$ of $M$.

Indeed, the reconstruction of the 3-body motion $\gamma(t)$ from the curve
$\bar{\gamma}(t)$ is a purely geometric problem, essentially the lifting
problem of a fiber bundle $M\rightarrow$ $\bar{M}$. Generally, the trajectory
is uniquely determined, modulo a global congruence, by the moduli curve, but
in this paper we shall not be concerned with the associated lifting procedure.

In the second step of the reduction procedure, $\bar{M}$ is replaced by a
subspace $M^{\ast}$ called the \emph{shape space}, consisting of congruence
classes of m-triangles of unit size, and hence the points of $M^{\ast}$
represent the (nonzero) similarity classes of m-triangles. It is a crucial
fact that $\bar{M}$ naturally identifies with the 3-space $\mathbb{R}^{3}$ and
with $M^{\ast}$ as its (unit) sphere $S^{2}$, with respect to the naturally
induced metric (see below). The projected image $\gamma^{\ast}(t)$ of the
moduli curve $\bar{\gamma}(t)$ in $M^{\ast}$ is referred to as the \emph{shape
curve}. By a \emph{geometric curve} on the sphere we shall mean a curve
parametrized by arc-length (or with no specific parametrization).

Along these lines, in a previous paper \cite{HS-2} the authors have
investigated subtle questions pertaining to 3-body motions with vanishing
angular momentum, which should be regarded as Part I of the present study. In
Part I we showed that the moduli curves $\bar{\gamma}(t)$ representing 3-body
motions of vanishing angular momentum are the geodesics of the induced
Jacobi-metric on the moduli space $\bar{M}$, or equivalently, they are also
the solutions of the reduced Lagrange's least action principle. Among the
major results from Part I we mention the following characteristics for the
case of zero angular momentum:

\begin{itemize}
\item The moduli curve $\bar{\gamma}(t)$ of a 3-body motion is determined by
the associated shape curve $\gamma^{\ast}(t)$ on the 2-sphere.

\item The unique parametrization theorem asserts that the time parametrized
curve $\bar{\gamma}(t)$ is (essentially) determined by the oriented geometric
shape curve $\gamma^{\ast}$. In turn, the latter is uniquely determined by a
few curvature invariants representing the relative geometry between
$\gamma^{\ast}$ and the gradient flow of the potential function $U^{\ast}$ on
the 2-sphere at a
\ \ \ \ \ \ \ \ \ \ \ \ \ \ \ \ \ \ \ \ \ \ \ \ \ \ \ \ \ \ \ \ \ \ \ \ \ \ \ \ \ \ \ \ \ \ \ \ \ \ \ \ \ \ \ \ \ \ \ \ \ \ \ \ \ \ \ \ \ \ \ \ \ \ \ \ \ \ \ \ \ \ \ \ \ \ \ \ \ \ \ \ \ \ \ \ \ \ \ \ \ \ \ \ \ \ \ \ \ \ \ \ \ \ \ \ \ \ \ \ \ \ \ \ \ \ \ \ \ \ \ \ \ \ \ \ \ \ \ \ \ \ \ \ \ \ \ \ \ \ \ \ \ \ \ \ \ \ \ \ \ \ \ \ \ \ \ \ \ \ \ \ \ \ \ \ \ \ \ \ \ \ \ \ \ \ \ \ \ \ \ \ \ \ \ \ \ \ \ \ \ \ \ \ \ \ \ \ \ \ \ \ \ \ \ \ \ \ \ \ \ \ \ \ \ \ \ \ \ \ \ \ \ \ \ \ \ \ \ \ \ \ \ \ \ \ \ \ \ \ \ \ \ \ \ \ \ \ \ \ \ \ \ \ \ \ \ \ \ \ \ \ \ \ \ \ \ \ \ \ \ \ \ \ \ \ \ \ \ \ \ \ \ \ \ \ \ \ \ \ \ \ \ \ \ \ \ \ \ \ \ \ \ \ \ \ \ \ \ \ \ \ \ \ \ \ \ \ \ \ \ \ \ \ \ \ \ \ \ \ \ \ \ \ \ \ \ \ \ \ \ \ \ \ \ \ \ \ \ \ \ \ \ \ \ \ \ \ \ \ \ \ \ \ \ \ \ \ \ \ \ \ \ \ \ \ \ \ \ \ \ \ \ \ \ \ \ \ \ \ \ \ \ \ \ \ \ \ \ \ \ \ \ \ \ \ \ \ \ \ \ \ \ \ \ \ \ \ \ \ \ \ \ \ \ \ \ \ \ \ \ \ \ \ \ \ \ \ \ \ \ \ \ \ \ \ \ \ \ \ \ \ \ \ \ \ \ \ \ \ \ \ \ \ \ \ \ \ \ \ \ \ \ \ \ \ \ \ \ \ \ \ \ \ \ \ \ \ \ \ \ \ \ \ \ \ \ \ \ \ \ \ \ \ \ \ \ \ \ \ \ \ \ \ \ \ \ \ \ \ \ \ \ \ \ \ \ \ \ \ \ \ \ \ \ \ \ \ \ \ \ \ \ \ \ \ \ \ \ \ \ \ \ \ \ \ \ \ \ \ \ \ \ \ \ \ \ \ \ \ \ \ \ \ generic
point.

\item A remarkable property of the shape curve on the 2-sphere is expressed by
The Monotonicity Theorem, describing the piecewise monotonic behavior of its
(mass-modified) latitude between two succeeding local maxima or minima, which
must lie on different hemispheres. \ 
\end{itemize}

The purpose of this paper is to establish the first two of the above stated
properties to planar motions in general. On the other hand, the monotonicity
theorem is no longer valid.

We start with a description in the next two subsections of the geometric
reduction procedure, following the basic setting from \cite{HS-2}, and a
summary with the major results, Theorem A and Theorem B, is presented in
Section 1.3. In Section 2 we present the reduced Newtonian ODE system, at the
level of the moduli space $\bar{M}\simeq\mathbb{R}^{3}$. We shall also point
out the subtle distinction between this system of differential equations and
the geodesic equations of the dynamical Riemannian metric on $\bar{M}$
associated with a possible reduction applied to Jacobi's geometrization
approach. The two systems are identical if and only if the angular momentum
vanishes. But they are distinguished by some terms depending linearly on the
angular momentum, representing the effect of a ficticious "Coriolis force" as
if the sphere $M^{\ast}$ is rotating.

In Section 3 we show how a few geometric invariants of the shape curve
$\gamma^{\ast}$, in general at a single regular point, provide enough
information to determine the initial data for the moduli curve $\bar{\gamma
}(t)$ as the unique solution of the reduced ODE system. This also completes
the proofs of theorems stated in Section 1.3.

\ 

\subsection{The basic kinematic quantities and the potential function}

The classical 3-body problem in celestial mechanics studies the local and
global geometry of the trajectories of a 3-body system, namely the motion of
three point masses (bodies) of mass $m_{i}>0,i=1,2,3$, under the influence of
the mutual gravitational forces. This system constitutes a conservative
mechanical system with the Newtonian potential function
\begin{equation}
U=\sum_{i<j}\frac{m_{i}m_{j}}{r_{ij}},\text{ \ \ \ }r_{ij}=\left\vert
\mathbf{a}_{i}-\mathbf{a}_{j}\right\vert \text{, }\mathbf{a}_{i}\in
\mathbb{R}^{3}, \label{1.1}%
\end{equation}
and potential energy $-U$. We introduce the notion of an \emph{m-triangle,}
which we shall identify with the vector $\delta=(\mathbf{a}_{1},\mathbf{a}%
_{2},\mathbf{a}_{3})\in$ $\mathbb{R}^{3\times3}$ which records the position of
the system in a barycentric inertial frame, namely we also assume $\sum
m_{i}\mathbf{a}_{i}=0$.

A \emph{trajectory} is a time parametrized curve $\gamma(t)=(\mathbf{a}%
_{1}(t),\mathbf{a}_{2}(t),\mathbf{a}_{3}(t))\mathbf{\ }$representing a motion
of the 3-body system, locally characterized by Newton's equation%
\begin{equation}
\frac{d^{2}}{dt^{2}}\gamma=\nabla U(\gamma)=(\frac{1}{m_{1}}\frac{\partial
U}{\partial\mathbf{a}_{1}},\frac{1}{m_{2}}\frac{\partial U}{\partial
\mathbf{a}_{2}},\frac{1}{m_{3}}\frac{\partial U}{\partial\mathbf{a}_{3}%
}\text{\ })\text{\ } \label{1.2}%
\end{equation}
However, the trajectories can also be characterized globally as solutions of a
suitable boundary value problem, namely as extremals of an appropriate least
action principle, such as the two principles due to Lagrange and Hamilton.

Let us also recall the basic kinematic quantities which are the (polar) moment
of inertia, kinetic energy and angular momentum, respectively defined by
\begin{equation}
I=\sum m_{i}\left\vert \mathbf{a}_{i}\right\vert ^{2}\text{, \ }T=\frac{1}%
{2}\sum m_{i}\left\vert \mathbf{\dot{a}}_{i}\right\vert ^{2}\text{,
\ \ \ }\mathbf{\Omega}=\sum m_{i}(\mathbf{a}_{i}\times\mathbf{\dot{a}}_{i})
\label{1.3}%
\end{equation}
The dynamics of the 3-body problem is largely expressed by their interactions
with the potential function $U$ \ via the equation (\ref{1.2}), and the
invariance of the total energy
\begin{equation}
h=T-U \label{h}%
\end{equation}
is a simple consequence of (\ref{1.2}) and the definition of $T$, whereas the
invariance of the vector $\mathbf{\Omega}$ also follows from (\ref{1.2}), but
is essentially due to the rotational symmetry of $U$.

\subsection{Reduction to the moduli space and the shape space}

In this article we shall be concerned with planar three-body motions, namely
the individual position vectors $\mathbf{a}_{i}$ are confined to a fixed plane
$\mathbb{R}^{2}$ and hence the trajectory $\gamma(t)$ is a curve in the
\emph{configuration space}
\begin{equation}
M\simeq\mathbb{R}^{4}:\sum\limits_{i=1}^{3}m_{i}\mathbf{a}_{i}=0\text{,
\ }\mathbf{a}_{i}\in\mathbb{R}^{2} \label{M}%
\end{equation}
We assume the plane $\mathbb{R}^{2}$ is positively oriented by the unit normal
vector $\mathbf{k}$, and the angular momentum of a trajectory $\gamma(t)$ is
written as
\begin{equation}
\mathbf{\Omega}=\omega\mathbf{k}\text{,} \label{angmom}%
\end{equation}
so that the scalar angular momentum $\omega\in\mathbb{R}$ is a constant of the
motion. Thus, each trajectory $\gamma(t)$ belongs to a specific
energy-momentum level $(h,\omega)$. Consequently, due to the conservation of
energy and angular momentum, the Newtonian system (\ref{1.2}) for planar
motions reduces to a system of differential equations of total order $8-2=6$.

We assume $M$ has the Euclidean \emph{kinematic metric }$M$, with the inner
product of m-triangles $\delta=(\mathbf{a}_{1},\mathbf{a}_{2},\mathbf{a}%
_{3}),$ $\delta^{\prime}=(\mathbf{b}_{1},\mathbf{b}_{2},\mathbf{b}_{3})$
defined by
\begin{equation}
\delta\cdot\delta^{\prime}=\sum m_{i}\mathbf{a}_{i}\cdot\mathbf{b}_{i},
\label{metric}%
\end{equation}
and then the right side of Newton's equation (\ref{1.2}) is the gradient field
$\nabla U$. Moreover, the squared norm is the moment of inertia,
$I=I(\delta)=\left\vert \delta\right\vert ^{2}$, and the \emph{hyperradius
}$\rho=\sqrt{I}$\ is the natural size function\emph{ }which also measures the
distance from the origin.\emph{\ }

The linear group $SO(2)$ acts orthogonally on $M$ by rotating m-triangles, and
the orbit space of $M$ and its unit sphere $M^{1}\ $
\begin{equation}
\bar{M}=M/SO(2)\text{, \ }M^{\ast}=M^{1}/SO(2) \label{orbitspace}%
\end{equation}
are the (congruence) \emph{moduli space }and the \emph{shape space},
respectively. The points in $\bar{M}$ represent \emph{congruence classes}
$\bar{\delta}$ of m-triangles $\delta$, and points in $M^{\ast}$ represent the
\emph{shapes} (or similarity classes) $\delta^{\ast}$ of m-triangles
$\delta\neq0$.

Next, we recall that the above orbit spaces and related orbit maps is actually
the classical Hopf map construction $\mathfrak{h}:\mathbb{R}^{4}%
\rightarrow\mathbb{R}^{3}$ in disguise, which is illustrated by the following diagram%

\begin{equation}%
\begin{array}
[c]{ccccccc}%
M & \simeq & \mathbb{R}^{4} & ^{\mathfrak{h}}\rightarrow & \mathbb{R}^{3} &
\simeq & \bar{M}\\
\cup &  & \cup &  & \cup &  & \cup\\
M^{1} & \simeq & S^{3} & ^{\mathfrak{h}}\rightarrow & S^{2} & \simeq &
M^{\ast}%
\end{array}
\label{diagram}%
\end{equation}
where $M\simeq\mathbb{R}^{4}=$ $\mathbb{R}^{2\times2}$ is a chosen
$SO(2)$-equivariant isometry between $M$ and the matrix space $\mathbb{R}%
^{2\times2}$, whose column vectors $\mathbf{x}_{1},\mathbf{x}_{2}$ provides a
choice of coordinates for m-triangles and are also referred to as Jacobi
vectors in the literature. In the sequel we shall identify the pair $(\bar
{M},M^{\ast})$ with $(\mathbb{R}^{3},S^{2})$.

Similar to the pair $(M,M^{1})$, $\bar{M}$ is a cone over $M^{\ast}$ with
$\rho$ as the radial coordinate, and the sphere $M^{\ast}$ is the subset
$(\rho=1)$. These spaces have the naturally induced \emph{orbital distance
metric}, making the Hopf map a Riemannian submersion. With this geometry
$M^{\ast}=S^{2}(1/2)$ is the round sphere of radius $1/2$, namely with the
metric
\begin{equation}
d\sigma^{2}=\frac{1}{4}(d\varphi^{2}+\sin^{2}(\varphi)d\theta^{2}),
\label{MstarMetric}%
\end{equation}
in terms of spherical polar coordinates $(\varphi,\theta)$ on $S^{2}$, whereas
$(\bar{M},d\bar{s}^{2})$ is best understood as a Riemannian cone over
$M^{\ast}$:
\begin{equation}
\bar{M}=C(M^{\ast}):d\bar{s}^{2}=d\rho^{2}+\frac{\rho^{2}}{4}(d\varphi
^{2}+\sin^{2}(\varphi)d\theta^{2}) \label{MbarMetric}%
\end{equation}
Note, however, the representation of the various shapes of m-triangles on a
fixed model sphere $S^{2}$, with a distinguished equator circle representing
collinear shapes, depends on some choice of conventions together with the mass
distribution $\left\{  m_{i}\right\}  $, via the mass dependence of the Jacobi vectors.

Briefly, in this article we shall focus on the two-step reduction%
\begin{equation}
M\rightarrow\bar{M}\text{, \ }\bar{M}-\left\{  0\right\}  \text{\ }\rightarrow
M^{\ast}\text{, \ \ \ }\gamma(t)\rightarrow\bar{\gamma}(t)\text{\ }%
\rightarrow\gamma^{\ast}(t)
\end{equation}
by which a trajectory $\gamma(t)$ of a planary 3-body motion is projected to
its moduli curve $\bar{\gamma}(t)$ and further to its shape curve
$\gamma^{\ast}(t)$ on the 2-sphere. In Section 2.1 we shall put the above
reduction in the framework of Riemannian geometry, and the relative geometry
of the shape curve and the gradient flow of the function $U^{\ast}$, the
restriction of $U$ to the sphere, will be our primary concern. For basic
information on this geometric reduction approach we refer to \cite{HS-1},
\cite{HS-2}.

\subsection{A summary of the main results}

The Hopf map construction (\ref{diagram}) makes it convenient to use a
Euclidean model, $\bar{M}$ $=\mathbb{R}^{3}$, for the moduli space with the
unit sphere $S^{2}(1)$ as the shape space $M^{\ast}.$ In this way one can
express all kinematic quantities and dynamical equations in terms of the usual
spherical geometry, and hence take the full advantage of the cone structure of
$\bar{M}$ over $M^{\ast}$ using the coordinates $(\rho,\varphi,\theta)$ on
$\bar{M}$, $0\leq\varphi\leq\pi,0\leq\theta\leq2\pi.$

Besides the $SO(2)$-invariance of the Newtonian potential function $U$,
another crucial property of $U$ that we shall exploit is its
\emph{homogeneity}, namely it is of type
\begin{equation}
U=\frac{U^{\ast}(\varphi,\theta)}{\rho} \label{S2}%
\end{equation}
For a given curve $\gamma(t)$ in $M$, the two curves%

\begin{equation}
\gamma^{\ast}(t)=(\varphi(t),\theta(t))\text{, \ \ }\bar{\gamma}%
(t)=(\rho(t),\gamma^{\ast}(t)) \label{S1}%
\end{equation}
are the associated shape and moduli curve, respectively. Let us consider
trajectories $\gamma(t)$ of Newton's equation (\ref{1.2}) at a given
energy-momentum level $(h,\omega)$. In Section 2.2 we shall focus on the
\emph{reduced Newton's equations} in $\bar{M}$ in the coordinates
$(\rho,\varphi,\theta)$, see (\ref{NewtonRedu}), which by spherical geometry
can be presented as the pair
\begin{align}
0  &  =\ddot{\rho}+\frac{\dot{\rho}^{2}}{\rho}-\frac{1}{\rho}(\frac{1}{\rho
}U^{\ast}+2h)\text{ \ }\nonumber\\
0  &  =\text{\ }\ddot{\gamma}^{\ast}+P\dot{\gamma}^{\ast}+Q\mathbf{\nu}^{\ast
}+R\nabla U^{\ast}\text{\ \ \ } \label{S5}%
\end{align}
Here the first equation in (\ref{S5}) is the so-called \emph{Lagrange-Jacobi}
equation, which in terms of the moment of inertia $I=\rho^{2}$ expresses as%
\[
\ddot{I}=2U+4h.
\]
The second equation is a vector equation on the unit 2-sphere, expressing the
covariant acceleration $\ddot{\gamma}^{\ast}$ of the shape curve as a sum of
three "forces", $\mathbf{\nu}^{\ast}$ is the oriented unit normal of the
curve, and $\nabla U^{\ast}$ is the gradient field of $U^{\ast}$. The three
coefficients can be expressed as
\begin{equation}
P=2\frac{\dot{\rho}}{\rho},\text{ }Q=\frac{2\omega v}{\rho^{2}},\text{
\ \ }R=-\frac{4}{\rho^{3}},
\end{equation}
where $v=\left\vert \dot{\gamma}^{\ast}\right\vert $ is the speed of the shape
curve. \ \ \ \ \ \ \ \ \ \ \ \ \ \ \ \ \ \ \ \ \ \ \ \ \ \ \ \ \ \ \ \ \ \ \ \ \ \ \ \ \ \ \ \ \ \ \ \ \ \ \ \ \ \ \ \ \ \ \ 

The scalar version of (\ref{S5}) is stated as the system (\ref{NewtonRedu}).
We may imagine the component $Q\mathbf{\nu}^{\ast}$ in (\ref{S5}) to be the
\emph{Coriolis} \ "force" caused by some fictitious rotation of the shape
sphere, and we note that its magnitude is proportional to the speed as well as
the angular momentum. As an ODE system, (\ref{S5}) should be augment
(\ref{S5}) with the energy integral (\ref{h}), viewed as a first order
equation in $\bar{M}$%
\begin{equation}
\frac{1}{2}\dot{\rho}^{2}+\frac{\rho^{2}}{8}v^{2}+\frac{\omega^{2}}{2\rho^{2}%
}-\frac{U^{\ast}}{\rho}=h \label{S7a}%
\end{equation}
where $h$ is a constant, called the energy level. In fact, combined with
(\ref{S7a} ) any of the three scalar equations in (\ref{NewtonRedu}) can be
derived from the other two, so the total order of the system is actually $5$.

As a curve on the 2-sphere, the basic geometric invariant of the shape curve
$\gamma^{\ast}$ is its geodesic curvature function $K^{\ast}=K^{\ast}(s)$,
where $s$ is the arc-length. In general, crucial information of the 3-body
motion is encoded into this function, and our\ problem is rather to detect the
code and extract the hidden information in an appropriate way. \ \ \ \ \ \ \ \ \ \ \ \ \ \ \ \ \ \ \ \ \ \ \ \ \ \ \ \ \ \ \ \ \ \ \ \ \ \ \ \ \ \ \ \ \ \ \ \ \ \ \ \ \ \ \ \ \ \ \ \ \ \ \ \ \ \ \ \ \ \ \ \ \ \ \ \ \ \ \ \ \ \ \ \ \ \ \ \ \ \ \ \ \ \ \ \ \ \ \ \ \ \ \ \ \ \ \ \ \ \ \ \ \ \ \ \ \ \ \ \ \ \ \ \ \ \ \ \ \ \ \ \ \ \ \ \ \ \ \ \ \ \ \ \ \ \ \ \ \ \ \ \ \ \ \ \ \ \ \ \ \ \ \ \ \ \ \ \ \ \ \ \ \ \ \ \ \ \ \ \ \ \ \ \ \ \ \ \ \ \ \ \ \ \ \ \ \ \ \ \ \ \ \ \ \ \ \ \ \ \ \ \ \ \ \ \ \ \ \ \ \ \ \ \ \ \ \ \ \ \ \ \ \ \ \ \ \ \ \ \ \ \ \ \ \ \ \ \ \ \ \ \ \ \ \ \ \ \ \ \ \ \ \ \ \ \ \ \ \ \ \ \ \ \ \ \ \ \ \ \ \ \ \ \ \ \ \ \ \ \ \ \ \ \ \ \ \ \ \ \ \ \ \ \ \ \ \ \ \ \ \ \ \ \ \ \ \ \ \ \ \ \ \ \ \ \ \ \ \ \ \ \ \ \ \ \ \ \ \ \ \ \ \ \ \ \ \ \ \ \ \ \ \ \ \ \ \ \ \ \ \ \ \ \ \ \ \ \ \ \ \ \ \ \ \ \ \ \ \ \ \ \ \ \ \ \ \ \ \ \ \ \ \ \ \ \ \ \ \ \ \ \ \ \ \ \ \ \ \ \ \ \ \ \ 

Let $U_{\mathbf{\nu}}^{\ast}$ denote the\ normal derivative of the function
$U^{\ast}$ along the curve $\gamma^{\ast}$. We shall derive the following
formula for the curvature
\begin{equation}
K^{\ast}=4\frac{U_{\mathbf{\nu}}^{\ast}}{\rho^{3}v^{2}}-\frac{2\omega}%
{\rho^{2}v} \label{K*}%
\end{equation}
This identity is the key to the understanding of how the relative geometry
between $\gamma^{\ast}$ and the gradient flow of $U^{\ast}$, in fact, provides
the data for the initial value problem of the ODE (\ref{S5}) and thus
determine the moduli curve \ $\bar{\gamma}(t)$. This is the major issue we
shall be addressing, and the main results can be formulated neatly as the
following two theorems :

\textbf{Theorem A \ }\emph{For a given total energy and angular momentum, a
planar three-body motion is completely determined up to congruence by its time
parametrized shape curve (which records only the changing of \ the similarity
class).}\ \ \ 

\textbf{Theorem B}\emph{ (Elimination of time) The time parametrized shape
curve is determined by the oriented geometric (i.e. non-parametrized) shape
curve. }\ \ \ \ \ \ \ \ \ \ \ \ \ \ \ \ \ \ \ \ \ \ \ 

\begin{remark}
(i) The main purpose of the present paper is to give the complete proofs of
the two theorems.

(ii) The theorems and their proof do not apply to the case of exceptional
shape curves (as defined in Section 2.4.4). Uniqueness of time parametrization
means, of course, unique modulo time translation. \ 

(iii) The simpler case of vanishing angular momentum $(\omega=0)$ was treated
in the paper \cite{HS-2}, where the same two theorems are proved, with special
attention to the case of $(h,\omega)=(0,0)$.

(iv) Planar 3-body motions were also investigated in \cite{HS-3}, attempting
to generalize results from \cite{HS-2}. However, in \cite{HS-3} the proof of
the result corresponding to the above two theorems is incorrect when
$\omega\neq0$, since the Coriolis term $Q\mathbf{\nu}^{\ast}$ in (\ref{S5})
was missing; we refer to the discussion in Section 2.1 and 2.2.
\end{remark}

\section{Geometric reduction}

\subsection{Riemannian structures on the moduli space}

In his famous lectures \cite{Jac}, Jacobi introduced the concept of a
\emph{kinematic metric }$ds^{2}$ on the configuration space $M$ of a
mechanical system with kinetic energy $T.$ For example, in the case of an
n-body system with total mass $\sum m_{i}=1,$\emph{\ }
\begin{equation}
ds^{2}=2Tdt^{2}=\sum\limits_{i}m_{i}(dx_{i}^{2}+dy_{i}^{2}+dz_{i}^{2})
\label{2.1}%
\end{equation}
which is clearly equivalent to the definition (\ref{metric}). Now, for a
system with potential energy $-U$ and a fixed total energy $h$, set
\begin{align}
M_{h}  &  =\left\{  p\in M;h+U(p)\geq0\right\} \label{2.3}\\
ds_{h}^{2}  &  =(h+U)ds^{2}\nonumber
\end{align}
where $ds_{h}^{2}$ will be referred to as the \emph{dynamical metric} on
$M_{h}$. By writing
\[
ds_{h}=\sqrt{h+U}ds=\sqrt{T}ds=\sqrt{2}Tdt
\]
Jacobi transformed Lagrange's action integral (on the left side of
(\ref{2.4})) into an arc-length integral, namely
\begin{equation}
J(\gamma)=\int_{\gamma}Tdt=\frac{1}{\sqrt{2}}\int_{\gamma}ds_{h} \label{2.4}%
\end{equation}
and hence the least action principle becomes the following simple geometric
statement :
\begin{align}
&  \text{\emph{Trajectories} \emph{with total energy} }\emph{h}\text{
\emph{are along the} \emph{geodesic curves}}\label{state1}\\
\text{ }  &  \text{\emph{in the space} }M_{h}\text{ \emph{with the dynamical
metric} }ds_{h}^{2\text{ }}.\nonumber
\end{align}

Nowadays, the metric spaces $(M,ds^{2}),(M_{h},ds_{h}^{2\text{ }})$ are called
Riemannian manifolds, and the dynamical metric is a conformal modification of
the kinematic metric by the scaling factor $(h+U)$. In general, and as
exemplified by (\ref{2.1}), a Riemannian metric on a manifold $N$ amounts to
the choice of a kinetic energy function on the tangent bundle,
$T:TN\rightarrow\mathbb{R}$, which is a positive definite quadratic form on
each tangent plane $T_{p}N$. This allows us to define the speed $\left\vert
\frac{d\Gamma}{dt}\right\vert $ along a time parametrized curve $\Gamma(t)$ in
$N$ and hence an arc-length function $u(t)$ along the curve by
\begin{equation}
T(\frac{d\Gamma}{dt})=\frac{1}{2}\left\vert \frac{d\Gamma}{dt}\right\vert
^{2}=\frac{1}{2}(\frac{du}{dt})^{2} \label{2.7}%
\end{equation}

Next, we would like to inquire further into the above Lagrange-Jacobi approach
to dynamics, to check whether the dynamics in $M$ and characterization
(\ref{state1}) of the trajectories can be pushed down to $\bar{M}$, with a
similar geometric description of the moduli curves of 3-body motions on a
fixed energy-momentum level $(h,\omega)$.

First of all, $\bar{M}$ already has the orbital distance metric $d\bar{s}^{2}$
as in (\ref{MbarMetric}) and hence a corresponding notion of kinetic energy
$\bar{T}$ as indicated in (\ref{2.7}), namely
\begin{equation}
d\bar{s}^{2}=2\bar{T}dt^{2}=d\rho^{2}+\rho^{2}d\sigma^{2}=d\rho^{2}+\frac
{\rho^{2}}{4}(d\varphi^{2}+\sin^{2}(\varphi)d\theta^{2})\text{ } \label{Tbar}%
\end{equation}
On the other hand, for a curve $\gamma(t)$ in $M$ there is the orthogonal
decomposition $\dot{\gamma}=\dot{\gamma}^{h}+\dot{\gamma}^{\omega}$ of its
velocity, and the corresponding splitting of kinetic energy%
\begin{equation}
T=\frac{1}{2}|\dot{\gamma}^{h}|^{2}+\frac{1}{2}|\dot{\gamma}^{\omega}%
|^{2}=T^{h}+T^{\omega}, \label{2.10}%
\end{equation}
where $\dot{\gamma}^{\omega}$ is tangential to the $SO(2)$-orbit. Hence,
$T^{\omega}$ is the kinetic energy due to purely rotational motion of the
m-triangle, and for planar 3-body motions this energy term can be expressed
as
\begin{equation}
T^{\omega}=\frac{\omega^{2}}{2\rho^{2}}. \label{Tomega}%
\end{equation}

By definition of the metric $d\bar{s}^{2}$, the orbit map $M\rightarrow$
$\bar{M}$ is a Riemannian submersion and hence maps the "horizontal" velocity
$\dot{\gamma}^{h}$ of $\dot{\gamma}$ isometrically to the velocity vector of
$\bar{\gamma}$. This shows $T^{h}=\bar{T}$ is naturally the kinetic energy at
the level of $\bar{M}$, that is, the kinematic metric $d\breve{s}^{2}$ on
$\bar{M}$ naturally identifies with the orbital distance metric. Therefore, by
(\ref{2.10}) and (\ref{Tomega}), the kinematic metric on $\bar{M}$ can be
finally expressed as in (\ref{Tbar}):
\begin{align}
d\breve{s}^{2}  &  =2T^{h}dt^{2}=2(T-T^{\omega})dt^{2}=2(T-\frac{\omega^{2}%
}{2\rho^{2}})dt^{2}\label{2.11}\\
&  =d\bar{s}^{2}=2\bar{T}dt^{2}=d\rho^{2}+\frac{\rho^{2}}{4}(d\varphi^{2}%
+\sin^{2}\varphi d\theta^{2})\nonumber
\end{align}

Now, in the spirit of Lagrange-Jacobi, let us turn to the description of
$\bar{M}$ as the configuration space of a simple classical mechanical system,
with the kinetic energy $\bar{T}$ of the kinematic metric as above, and
\emph{effective} \emph{potential} energy $-\bar{U}$ defined so that the system
is conservative, namely by setting
\begin{equation}
\bar{U}=U-\frac{\omega^{2}}{2\rho^{2}}\text{, \ \ }\bar{T}-\bar{U}%
=T-U=h\text{\ } \label{U*red}%
\end{equation}
Then the Lagrange function is $\bar{L}=\bar{T}+\bar{U}$, and according to
Lagrange's least action principle the trajectories of the mechanical system
should be the solutions of Lagrange's equations%

\begin{equation}
\frac{d}{dt}(\frac{\partial\bar{L}}{\partial\dot{\rho}})=\frac{\partial\bar
{L}}{\partial\rho}\text{, \ \ }\frac{d}{dt}(\frac{\partial\bar{L}}%
{\partial\dot{\varphi}})=\frac{\partial\bar{L}}{\partial\varphi},\text{
\ }\frac{d}{dt}(\frac{\partial\bar{L}}{\partial\dot{\theta}})=\frac
{\partial\bar{L}}{\partial\theta} \label{L-eq}%
\end{equation}
\qquad

Alternatively, Lagrange's approach may well be modified according to Jacobi's
geometrization idea, leading to the dynamical Riemannian metric $d\bar
{s}_{h,\omega}^{2}:$
\begin{align}
\bar{J}(\bar{\gamma})  &  =\sqrt{2}\int_{\bar{\gamma}}\bar{T}dt=\sqrt{2}%
\int_{\bar{\gamma}}(\bar{U}+h)dt=\int_{\bar{\gamma}}\sqrt{\bar{U}+h}d\bar
{s}=\int d\bar{s}_{h,\omega}\text{ ,}\nonumber\\
d\bar{s}_{h,\omega}^{2}  &  =\bar{T}d\bar{s}^{2}=(\bar{U}+h)d\bar{s}%
^{2}=(U+h-\frac{\omega^{2}}{2\rho^{2}})d\bar{s}^{2} \label{2.12}%
\end{align}

In summary, the subregion $\bar{M}_{h,\omega}$ $\subset\bar{M}$ defined by
$\bar{U}+h\geq0$ can be regarded as a classical mechanical system, with
kinematic metric $d\bar{s}^{2}=2\bar{T}dt^{2}$, potential function $\bar{U}$,
and energy conservation $h=$ $\bar{T}-\bar{U}$. But it is also a Riemannian
manifold with the dynamical metric $d\bar{s}_{h,\omega}^{2}=(\bar{U}%
+h)d\bar{s}^{2}$. Thus we arrive at the following geometric statement similar
to (\ref{state1}) :%
\begin{align}
&  \text{\emph{Trajectories of the simple mechanical system }\ on }\bar
{M}_{h,\omega}\text{ \emph{are the} \emph{solutions }}\nonumber\\
&  \text{\emph{of the Lagrange equations} (\ref{L-eq})\emph{, and the curves
coincide with} }\label{State2}\\
&  \text{\emph{the geodesic curves of the} \emph{dynamical metric} }d\bar
{s}_{h,\omega}^{2}\text{\emph{.}}\nonumber
\end{align}
In terms of the coordinates $(\rho,\varphi,\theta)$ on $\bar{M}$, both ways of
calculating the associated differential equations lead to the following ODE system%

\begin{align}
(i)\text{ \ }0  &  =\ddot{\rho}+\frac{\dot{\rho}^{2}}{\rho}-\frac{1}{\rho
}(\frac{1}{\rho}U^{\ast}+2h)\nonumber\\
(ii)\text{ \ }0  &  =\text{\ }\ddot{\varphi}+2\frac{\dot{\rho}}{\rho}%
\dot{\varphi}-\frac{1}{2}\sin(2\varphi)\dot{\theta}^{2}-\frac{4}{\rho^{3}%
}U_{\varphi}^{\ast}\label{fake}\\
(iii)\text{ \ \ }0  &  =\ddot{\theta}+2\frac{\dot{\rho}}{\rho}\dot{\theta
}+2\cot(\varphi)\dot{\varphi}\dot{\theta}-\frac{4}{\rho^{3}}\frac{1}{\sin
^{2}\varphi}U_{\theta}^{\ast}\nonumber
\end{align}
where $U_{\varphi}^{\ast}=\frac{\partial}{\partial\varphi}U^{\ast}$ etc. , and
the following first integral is the energy conservation law
\begin{equation}
(iv)\text{ \ }h=\bar{T}-\bar{U}=\frac{1}{2}\dot{\rho}^{2}+\frac{\rho^{2}}%
{8}[\dot{\varphi}^{2}+(\sin^{2}\varphi)\dot{\theta}^{2}]+\frac{\omega^{2}%
}{2\rho^{2}}-\frac{U^{\ast}}{\rho} \label{h-int}%
\end{equation}
Therefore, the system (i)-(iv) of differential equations has total order $5$.

\begin{remark}
The above system was derived in the same way in \cite{HS-2} for the special
case of vanishing angular momentum $(\omega=0)$, and then the equations
actually yield the moduli curves of the 3-body motions. However, this fails
when $\omega\neq0$, namely the ordinary Lagrange-Jacobi approach does not
yield the correct reduced Newton's equations (\ref{NewtonRedu}) at the level
of $\bar{M}$. The subtle difference between the latter and the above system
(\ref{fake}) is conspicuous by direct comparison, namely the Coriolis term is
missing in the system (\ref{fake}).
\end{remark}

\subsection{The reduced Newton's equations on the moduli space}

The induced Newton's equations on the moduli space $\bar{M}\simeq
\mathbb{R}^{3}$, for given values of $(h,\omega)$, can be expressed in
spherical coordinates $(\rho,\varphi,\theta)$ as the following ODE system:
\begin{align}
(i)\text{ \ }0  &  =\ddot{\rho}+\frac{\dot{\rho}^{2}}{\rho}-\frac{1}{\rho
}(\frac{1}{\rho}U^{\ast}+2h)\nonumber\\
(ii)\text{ \ }0  &  =\text{\ }\ddot{\varphi}+2\frac{\dot{\rho}}{\rho}%
\dot{\varphi}-\frac{1}{2}\sin(2\varphi)\dot{\theta}^{2}-2\omega\frac
{\sin\varphi}{\rho^{2}}\dot{\theta}-\frac{4}{\rho^{3}}U_{\varphi}^{\ast
}\label{NewtonRedu}\\
(iii)\text{ \ \ }0  &  =\ddot{\theta}+2\frac{\dot{\rho}}{\rho}\dot{\theta
}+2\cot(\varphi)\dot{\varphi}\dot{\theta}+2\omega\frac{1}{\rho^{2}\sin\varphi
}\dot{\varphi}-\frac{4}{\rho^{3}}\frac{1}{\sin^{2}\varphi}U_{\theta}^{\ast
}\nonumber
\end{align}
which also has the first integral (iv) stated in (\ref{h-int}), namely the
conservation of energy. Again, the total order of the system (i)-(iv) is 5.

\begin{remark}
These equations are valid for any shape potential function $U^{\ast}$, and the
results in this paper do not depend on specific properties of $U^{\ast}$. For
the Newtonian case there is the explicit expression (\ref{U*}) below.
\end{remark}

Clearly, for $\omega=0$ the two systems (\ref{fake}) and (\ref{NewtonRedu})
coincide, and the choice of spherical polar coordinates $(\varphi,\theta)$ on
the shape sphere $M^{\ast}=S^{2}$ is immaterial. However, for $\omega\neq0$
the explicit form of the Coriolis term in equation (ii) and (iii) depends on
the polar coordinates $(\varphi,\theta)$ to be centered at the north pole
$(\varphi=0)$, and then the equator circle $(\varphi=\pi/2)$ is the locus
representing the eclipse shapes. On this circle there are three distinguished
points $b_{i}$, of longitude angle $\theta_{i},i=1,2,3$, representing the
binary collisions, say $\theta_{1}$ represents the collision of point masses
$m_{2}$ and $m_{3}$ etc. Choosing the zero meridian to be $\theta_{1}=0$ say,
we shall assume positive direction of $\theta$ so that%
\[
(\theta_{1},\theta_{2},\theta_{3})=(0,\beta_{3},-\beta_{2}),\text{ \ }%
\cos\beta_{i}=\frac{\mu_{j}\mu_{k}-\mu_{i}}{(1-\mu_{j})(1-\mu_{k})}\text{ \ ,
cf. (133) in [4]}%
\]
where the angle $\beta_{i}$ is the longitude distance between $b_{j}$ and
$b_{k}$, for different $i,j,k$, and we have introduced the normalized masses
\[
\mu_{i}=m_{i}/\bar{m};\text{ }\bar{m}=\sum m_{i}%
\]
For convenience, the Newtonian shape potential function $U^{\ast}$ can be
expressed as (cf. (169) in [4])
\begin{equation}
U^{\ast}(\varphi,\theta)=\bar{m}^{5/2}\sum_{i=1}^{3}\frac{(\mu_{j}\mu
_{k})^{3/2}(\mu_{i}^{\ast})^{-1/2}}{\sqrt{(1-\sin\varphi\cos(\theta-\theta
_{i})}}\text{, \ }\mu_{i}^{\ast}=\frac{1}{2}(1-\mu_{i})\text{.} \label{U*}%
\end{equation}
For example, in the special case of $m_{1}=m_{2}=m_{3}=1$ we obtain
\begin{equation}
U^{\ast}(\varphi,\theta)=\sum_{i=1}^{3}\frac{1}{\sqrt{(1-\sin\varphi
\cos(\theta-\theta_{i})}} \label{U*1}%
\end{equation}

\bigskip For a derivation of the above ODE system, we shall recall the proof
given in Sydnes[2013], adapted to the planar case. It is based upon singular
value decomposition of $M$ as the space of matrices $X=[\mathbf{x}%
_{1}|\mathbf{x}_{2}]$ with column (Jacobi) vectors $\mathbf{x}_{i},$
\[
\Phi:\mathfrak{S}=SO(2)\times D\times SO(2)^{\prime}\rightarrow\mathbb{R}%
^{2\times2}\simeq M\text{,}%
\]
$D\simeq\mathbb{R}^{2}$ consists of diagonal matrices $diag(r_{1},r_{2})$, and
$\Phi$ is the surjective map given by matrix multiplication and is locally an
analytic diffeomorphism at generic points $(P,R,Q)\in\mathfrak{S}$. With the
following parametrization
\begin{align*}
P  &  =\left(
\begin{array}
[c]{cc}%
\cos\alpha & -\sin\alpha\\
\sin\alpha & \cos\alpha
\end{array}
\right)  ,\text{ }R=\rho\left(
\begin{array}
[c]{cc}%
\sin(\frac{\varphi}{2}+\frac{\pi}{4}) & 0\\
0 & \cos(\frac{\varphi}{2}+\frac{\pi}{4})
\end{array}
\right)  ,\text{ }\\
Q  &  =\left(
\begin{array}
[c]{cc}%
\cos\frac{\theta}{2} & \sin\frac{\theta}{2}\\
-\sin\frac{\theta}{2} & \cos\frac{\theta}{2}%
\end{array}
\right)  ,
\end{align*}
we can use $(\alpha,\rho,\varphi,\theta)$ as (local) coordinates in $M$, where
$\rho,\varphi,\theta$ have the previous geometric interpretation as
coordinates in the moduli space $\bar{M}=M/SO(2)$, and the angle $\alpha$
parametrizes the "congruence" rotation group $SO(2)$ acting by multiplication
on the left side of $\mathfrak{S}$. Moreover, the columns $\mathbf{u}%
_{1},\mathbf{u}_{2}$ of $P$ are the eigenvectors of the inertia tensor of the
3-body system, hence constitute an intrinsic moving frame of the 3-body
motion, and by definition of $\Phi$%
\[
\lbrack\mathbf{x}_{1}|\mathbf{x}_{2}]=\rho\lbrack\mathbf{u}_{1}|\mathbf{u}%
_{2}]\left(
\begin{array}
[c]{cc}%
\sin(\frac{\varphi}{2}+\frac{\pi}{4})\cos\frac{\theta}{2} & \sin(\frac
{\varphi}{2}+\frac{\pi}{4})\sin\frac{\theta}{2}\\
-\cos(\frac{\varphi}{2}+\frac{\pi}{4})\sin\frac{\theta}{2} & \cos
(\frac{\varphi}{2}+\frac{\pi}{4})\cos\frac{\theta}{2}%
\end{array}
\right)
\]

Now, the kinetic energy can be expressed as%
\begin{align*}
T  &  =\frac{1}{2}tr(\dot{X}\dot{X}^{t})=\frac{1}{2}(|\mathbf{\dot{x}}%
_{1}|^{2}+|\mathbf{\dot{x}}_{2}|^{2})\\
&  =\frac{1}{2}\dot{\rho}^{2}+\frac{\rho^{2}}{8}(\dot{\varphi}^{2}+\dot
{\theta}^{2})+\frac{1}{2}\rho^{2}\dot{\alpha}^{2}-\frac{1}{2}\rho^{2}%
(\cos\varphi)\dot{\alpha}\dot{\theta},
\end{align*}
and the angular momentum as
\begin{equation}
\omega=\frac{\partial T}{\partial\dot{\alpha}}=\rho^{2}(\dot{\alpha}-\frac
{1}{2}\cos\varphi\dot{\theta}), \label{alfadot}%
\end{equation}
which also equals the cross-product $X\times\dot{X}$ \ (appropriately defined).

The equations of the Newtonian motion $t\rightarrow X(t)$ are equivalent to
the Euler-Lagrange equations associated with the Lagrange function
\[
L=T+U=T+\frac{U^{\ast}(\varphi,\theta)}{\rho},\
\]
namely the following equations
\begin{align*}
\dot{\omega}  &  =\frac{d}{dt}(\frac{\partial T}{\partial\dot{\alpha}}%
)=\frac{d}{dt}(\frac{\partial L}{\partial\dot{\alpha}})=\frac{\partial
L}{\partial\alpha}=0\\
\frac{d}{dt}(\frac{\partial L}{\partial\dot{\rho}})  &  =\frac{\partial
L}{\partial\rho},\text{ }\frac{d}{dt}(\frac{\partial L}{\partial\dot{\varphi}%
})=\frac{\partial L}{\partial\varphi},\text{ }\frac{d}{dt}(\frac{\partial
L}{\partial\dot{\theta}})=\frac{\partial L}{\partial\theta}%
\end{align*}
The first equation simply says $\omega$ is constant. In the last three
equations the occurrence of $\dot{\alpha}$ and $\ddot{\alpha}$ can be
eliminated by (\ref{alfadot}), namely using%
\[
\dot{\alpha}=\frac{\omega}{\rho^{2}}+\frac{1}{2}\cos\varphi\dot{\theta}\text{,
}%
\]
and this yields, in fact, the three equations of (\ref{NewtonRedu}).

\subsection{The initial value problem in the moduli space}

The ODE (\ref{NewtonRedu}) in the moduli space $\bar{M}\simeq$ $\mathbb{R}%
^{3}$ is an analytic system which depends analytically on the scalar angular
momentum $\omega$, and moreover, there is a first integral (\ref{h-int}) whose
value at a given integral curve is the total energy $h$. Thus we shall refer
to the pair of constants $(h,\omega)$ as the energy-momentum level of the
curve. These curves, which we shall refer to as \textit{moduli curves}, are
therefore analytic curves
\[
t\rightarrow\bar{\gamma}(t)=(\rho(t),\varphi(t),\theta(t)),
\]
with power series expansions%
\begin{align}
\rho(t)  &  =\rho_{0}+\rho_{1}t+\rho_{2}t^{2}+.....+\nonumber\\
\varphi(t)  &  =\varphi_{0}+\varphi_{1}t+\varphi_{2}t^{2}%
+.....+\label{intcurve}\\
\theta(t)  &  =\theta_{0}+\theta_{1}t+\theta_{2}t^{2}+.....+\nonumber
\end{align}
at a chosen initial point $(\rho_{0,}\varphi_{0},\theta_{0})$. The initial
value problem in $\bar{M}$ is to determine the solution $\bar{\gamma}(t)$ from
a given \emph{initial data }set%
\begin{equation}
\rho_{0},\varphi_{0},\theta_{0};\rho_{1},\varphi_{1},\theta_{1}
\label{initial}%
\end{equation}
We shall regard $\omega$ as a given constant (or parameter) whereas the value
of $h$ is determined by $\omega$ and the initial data (\ref{initial}) by
evaluation of the expression (\ref{h-int}) at initial time $t=0$.

In contrast to this, the spherical initial data
\begin{equation}
\varphi_{0},\theta_{0};\varphi_{1},\theta_{1} \label{initial*}%
\end{equation}
would not suffice to determine the associated \textit{shape curve}%
\[
t\rightarrow\gamma^{\ast}(t)=(\varphi(t),\theta(t))
\]
on the sphere, since the curve is not the solution of a second order
differential equation on the sphere. Anyhow, elimination of $\rho$ and solving
an initial value problem purely on the sphere does not seem to be a tractable
approach for our purpose.\ \ \ \ \ \ \ \ \ \ \ \ \ \ \ \ \ \ \ \ \ \ \ \ \ \ \ \ \ \ \ \ \ \ \ \ \ \ \ \ \ \ \ \ \ \ \ \ \ \ \ \ \ \ \ \ \ \ \ \ \ \ \ \ \ \ \ \ \ \ \ \ \ \ \ \ \ \ \ \ \ \ \ \ \ \ \ \ \ \ \ \ \ \ \ \ \ \ \ \ \ \ \ \ \ \ \ \ \ \ \ \ \ \ \ \ \ \ \ \ \ \ \ \ \ \ \ \ \ \ \ \ \ \ \ \ \ \ \ \ \ \ \ \ \ \ \ \ \ \ \ \ \ \ \ \ \ \ \ \ \ \ \ \ \ \ \ \ \ \ \ \ \ \ \ \ \ \ \ \ \ \ \ \ \ \ \ \ \ \ \ \ \ \ \ \ \ \ \ \ \ \ \ \ \ \ \ \ \ \ \ \ \ \ \ \ \ \ \ \ \ \ \ \ \ \ \ \ \ \ \ \ \ \ \ \ \ \ \ \ \ \ \ \ \ \ \ \ \ \ \ \ \ \ \ \ \ \ \ \ \ \ \ \ \ \ \ \ \ \ \ \ \ \ \ \ \ \ \ \ \ \ \ \ 

Let us briefly consider the initial value problem for the system
(\ref{NewtonRedu}) and the dependence on the parameter $\omega$.

\begin{lemma}
Suppose there is a planar 3-body motion
\[
t\rightarrow\gamma(t)=(\mathbf{a}_{1}(t),\mathbf{a}_{2}(t),\mathbf{a}_{3}(t))
\]
in the configuration space (\ref{M}), starting at $\gamma(0)=\delta_{0}$ with
the initial velocity $\dot{\gamma}(0)=\delta_{1}$ and angular momentum vector
$\omega\mathbf{k}$, cf. (\ref{1.3}) and (\ref{angmom}). Let the corresponding
initial data for the moduli curve $\bar{\gamma}(t)$ in $\bar{M}$ be the
numbers in (\ref{initial}). Then for any number $\omega_{\rhd}$ there is a
planar 3-body motion in $M$ with angular momentum $\omega_{\triangleright
}\mathbf{k}$, whose moduli curve $\bar{\gamma}_{\triangleright}(t)$ has the
same initial data (\ref{initial}) as $\bar{\gamma}(t)$. Namely, the moduli
curves $\bar{\gamma}(t)$ and $\bar{\gamma}_{\triangleright}(t)$ are determined
by the ODE (\ref{NewtonRedu}) with the same initial data, but with parameters
$(h,\omega)$ and $(h_{\triangleright},\omega_{\triangleright})$ respectively,
where $h_{\triangleright}$ is the total energy (\ref{h-int}) determined by
$\omega_{\triangleright}$ and the above initial data.
\end{lemma}

\begin{proof}
The initial velocity of the original 3-body motion has an orthogonal
decomposition \ \ \ \
\[
\delta_{1}=\delta_{1}^{\tau}+\delta_{1}^{\omega},\text{ \ }\delta_{1}^{\omega
}=\mathbf{\tilde{\omega}\times}\delta_{0}=(\mathbf{\tilde{\omega}\times a}%
_{1}(0),\mathbf{\tilde{\omega}\times a}_{2}(0),\mathbf{\tilde{\omega}\times
a}_{3}(0))
\]
where $\mathbf{\tilde{\omega}}$ is the (instantaneous) rotational velocity
vector, which in the case of planar motions is $\mathbf{\tilde{\omega}=}%
\omega\rho_{0}^{-2}\mathbf{k}$. Evidently, we may freely modify the velocity
component $\delta_{1}^{\omega}$ by changing the scalar $\omega$
correspondingly, keeping $\delta_{0}$ and $\delta_{1}^{\tau}$ unchanged. But
on the other hand, the initial data (\ref{initial}) in the moduli space
$\bar{M}$ depend only on $(\delta_{0},\delta_{1}^{\tau})$, so this proves the lemma.
\end{proof}

\begin{remark}
(i) The various shape curves $\gamma_{+}^{\ast}$ with the same initial data
(\ref{initial*}) at a point $p\in M^{\ast}$ are uniquely distinguished by the
parameter $\omega$, or equivalently by their curvature $K^{\ast}$ at $p$, see
(\ref{K*}).\ \ \ \ $\ $\ \ \ \ \ \ \ \ \ \ \ \ \ \ \ \ \ \ \ \ \ \ \ \ \ \ \ \ \ \ \ \ \ \ \ \ \ \ \ \ \ \ \ \ \ \ \ \ \ \ \ \ \ \ \ \ \ \ \ \ \ \ \ \ \ \ \ \ \ \ \ \ \ \ \ \ \ \ \ \ \ \ \ \ \ \ \ \ \ \ \ \ \ \ \ \ \ \ \ \ \ \ \ \ \ \ \ \ \ \ \ \ \ \ \ \ \ \ \ \ \ \ \ \ \ \ \ \ \ \ \ \ \ \ \ \ \ \ \ \ \ \ \ \ \ \ \ \ \ \ \ \ \ \ \ \ \ \ \ \ \ \ \ \ \ \ \ \ \ \ \ \ \ \ \ \ \ \ \ \ \ \ \ \ \ \ \ \ \ \ \ \ \ \ \ \ \ \ \ \ \ \ \ \ \ \ \ \ \ \ \ \ \ \ \ \ \ \ \ \ \ \ \ \ \ \ \ \ \ \ \ \ \ \ \ \ \ \ \ \ \ \ \ \ \ \ \ \ \ \ \ \ \ \ \ \ \ \ \ \ \ \ \ \ \ \ \ \ \ \ \ \ \ \ \ \ \ \ \ \ \ \ \ \ \ \ \ \ \ \ \ \ \ \ \ \ \ \ \ \ \ \ \ \ \ \ \ \ \ \ \ \ \ \ \ \ \ \ \ \ \ \ \ \ \ \ \ \ \ \ \ \ \ \ \ \ \ \ \ \ \ \ \ \ \ \ \ \ \ \ \ \ \ \ \ \ \ \ \ \ \ \ \ \ \ \ \ \ \ \ \ \ \ \ \ \ \ \ \ \ \ \ \ \ \ \ \ \ \ \ \ \ \ \ \ \ \ \ \ \ \ \ \ \ \ \ \ \ \ \ \ \ \ \ \ \ \ \ \ \ \ \ \ \ \ \ \ \ \ \ \ \ \ \ \ \ \ \ \ \ \ \ \ \ \ \ \ \ \ \ \ \ \ \ \ \ \ \ \ \ \ \ \ \ \ \ \ \ \ \ \ \ \ \ \ \ \ \ \ \ \ \ \ \ \ \ \ \ \ \ \ \ \ \ \ \ \ \ \ \ \ \ \ \ \ \ \ \ \ \ \ \ \ \ \ \ \ \ \ \ \ \ \ \ \ \ \ \ \ \ \ \ \ \ \ \ \ \ \ \ \ \ \ \ \ \ \ \ \ \ \ \ \ \ \ \ \ \ \ \ \ \ \ \ \ \ \ \ \ \ \ \ \ \ \ \ \ \ \ \ \ \ \ \ \ \ \ \ \ \ \ \ \ \ \ \ \ \ \ \ \ \ \ \ \ \ \ \ \ \ \ \ \ \ \ \ \ \ \ \ \ \ \ \ \ \ \ \ \ \ \ \ \ \ \ \ \ \ \ \ \ \ \ \ \ \ \ \ \ \ \ \ \ \ \ \ \ \ \ \ \ \ \ \ \ \ \ \ \ \ \ \ \ \ \ \ \ \ \ \ \ \ \ \ \ \ \ \ \ \ \ \ \ \ \ \ \ \ \ \ \ \ \ \ \ \ \ \ \ \ \ \ \ \ \ \ \ \ \ \ \ \ \ \ \ \ \ \ \ \ \ \ \ \ \ \ \ \ \ \ \ \ \ \ \ \ \ \ \ \ \ \ \ \ \ \ \ \ \ \ \ \ \ \ \ \ \ \ \ \ \ \ \ \ \ \ \ \ \ \ \ \ \ \ \ \ \ \ \ \ \ \ \ \ \ \ \ \ \ \ \ \ \ \ \ \ \ \ \ \ \ \ \ \ \ \ \ \ \ \ \ \ \ \ \ \ \ \ \ \ \ \ \ \ \ \ \ \ \ \ \ \ \ \ \ \ \ \ \ \ \ \ \ \ \ \ \ \ \ \ \ \ \ \ \ \ \ \ \ \ \ \ \ \ \ \ \ \ \ \ \ \ \ \ \ \ \ \ \ \ \ \ \ \ \ \ \ \ \ \ \ \ \ \ \ \ \ \ \ \ \ \ \ \ \ \ \ \ \ \ \ \ \ \ \ \ \ \ \ \ \ \ \ \ \ \ \ \ \ \ \ \ \ \ \ \ \ \ \ \ \ \ \ \ \ \ \ \ \ \ \ \ \ \ \ \ \ \ \ \ \ \ \ \ \ \ \ \ \ \ \ \ \ \ \ \ \ \ \ \ \ \ \ \ \ \ \ \ \ \ \ \ \ \ \ \ \ \ \ \ \ \ \ \ \ \ \ \ \ \ \ \ \ \ \ \ \ \ \ \ \ \ \ \ \ \ \ \ \ \ \ \ \ \ \ \ \ \ \ \ \ \ \ \ \ \ \ \ \ \ \ \ \ \ \ \ \ \ \ \ \ \ \ \ \ \ \ \ \ \ \ \ \ \ 

(ii) The various moduli curves $\bar{\gamma}(t)$ with the same initial data
(\ref{initial}) are also distinguished by the parameter $\omega$. For
$\omega=0$ the curves are geodesics with respect to the dynamical metric
$d\bar{s}_{h,0}^{2}$, cf. (\ref{2.12}), so $\omega$ acts as a deformation
parameter on this family of curves. \ 
\end{remark}

\begin{remark}
\label{scaling}The 3-body motions in $M$ have a 1-parameter symmetry group
$\{\Phi_{k}$,$k\neq0\}$, namely time-size scaling symmetries $\Phi_{k}$, which
transform a trajectory $t\rightarrow\gamma(t)$ to the tractory \
\[
t\rightarrow\gamma^{(k)}(t)=k^{-2/3}\gamma(kt)
\]

\end{remark}

The effect of the transformation $\Phi_{k}$ on initial data in $\bar{M}$ and
parameters $(h,\omega)$ is as follows%
\begin{align}
(\rho_{0},\varphi_{0},\theta_{0};\rho_{1},\varphi_{1},\theta_{1})  &
\rightarrow(k^{-2/3}\rho_{0},\varphi_{0},\theta_{0};k^{1/3}\rho_{1}%
,k\varphi_{1},k\theta_{1})\text{; }\label{effect}\\
\text{\ }\omega &  \rightarrow k^{-1/3}\omega,h\rightarrow k^{2/3}h\nonumber
\end{align}
In particular, there is the time reversal transformation $\Phi_{-1}$ which
converts a 3-body motion $t\rightarrow\gamma(t)$ to the motion $t\rightarrow
\gamma^{(-1)}(t)=\gamma(-t)$ in the opposite direction. This changes the sign
of $\omega$, but the sign of $h$ is invariant. Also note that the quantity
$H=h\omega^{2}$ is an invariant of the symmetry group, and so is the
(unoriented, geometric) shape curve $\gamma^{\ast}$.

\subsection{Geometry on the shape sphere}

\subsubsection{Temporal and intrinsic invariants and their order}

We shall make effective usage of the fact that the shape space is the round
sphere $S^{2}$, and moreover, essential information about the moduli curve
$\bar{\gamma}(t)$ in $\bar{M}$ is encoded into the relative geometry between
the shape curve $\gamma^{\ast}$ and the gradient field $\nabla U^{\ast}$ on
$S^{2}$. Two types of quantities (also called invariants) are involved in this
interplay and we shall refer to them as being either \emph{temporal} or
\emph{intrinsic}\textit{.}

The temporal invariants are associated with $\bar{\gamma}(t)$ and
differentiation with respect to time $t$, whereas the intrinsic ones depend on
the geometric curve $\gamma^{\ast}$ or the relative geometry between this
curve and the gradient flow of $U^{\ast}$, which may involve differentiations
with respect to the arc-length of $\gamma^{\ast}$. So generally, we shall
define the \emph{order} of the invariant to be the highest order of
differentiations with respect to $t$ or $s$. For example, the coefficients
$\rho_{k},\varphi_{k},\theta_{k}$ of the expansions (\ref{intcurve}) are
temporal invariants of order k, but we shall regard $\varphi_{0},\theta_{0}$
as intrinsic since they simply specify the chosen initial point $p\in$
$\gamma^{\ast}$.

In a more general setting, let us start with a given time parametrized
analytic curve $\gamma^{\ast}(t)=(\varphi(t),\theta(t))$ on the sphere $S^{2}%
$, and let $s\geq0$ be its arc-length parameter measured from an initial point
$p=\gamma^{\ast}(0)=(\varphi_{0},\theta_{0})$. The linkage between the
parameters $t$ and $s$, which is a one-to-one correspondence in general, is
given by the speed function
\begin{equation}
v=v(t)=ds/dt\geq0 \label{v}%
\end{equation}
and corresponding differential operators are related by
\begin{equation}
\frac{d}{ds}=\frac{1}{v}\frac{d}{dt}\text{, \ }\frac{d^{2}}{ds^{2}}=\frac
{1}{v^{2}}\frac{d^{2}}{dt^{2}}-\frac{\dot{v}}{v^{3}}\frac{d}{dt}\text{, etc.}
\label{t-s}%
\end{equation}
Since the sphere has the Riemannian metric $ds^{2}=d\varphi^{2}+(\sin
^{2}\varphi)d\theta^{2}$, the speed and its time derivative have the expressions%

\begin{equation}
v=\sqrt{\dot{\varphi}^{2}+(\sin^{2}\varphi)\dot{\theta}^{2}}\text{, \ }\dot
{v}=\frac{d}{dt}v=\frac{1}{v}[\dot{\varphi}\ddot{\varphi}+(\sin\varphi
\cos\varphi)\dot{\varphi}\dot{\theta}^{2}+\sin^{2}(\varphi)\dot{\theta}%
\ddot{\theta}] \label{3.2}%
\end{equation}
and they are viewed as temporal invariants of order 1 and 2, respectively.

\begin{definition}
A \emph{direction element} on the sphere $S^{2}$ consists of a point $p\in
S^{2}$ together with a tangential direction at $p$, and it is denoted by
$(J_{\varphi},J_{\theta})_{p}$ or simply $(J_{\varphi},J_{\theta})$ when the
point $p$ is tacitly understood. It is said to be \emph{regular} if $\nabla
U^{\ast}$ is transversal to it.
\end{definition}

In fact, the direction is represented by the pair
\[
(J_{\varphi},J_{\theta})=(\frac{\partial\varphi}{\partial s},\frac
{\partial\theta}{\partial s})
\]
which is an intrinsic invariant of order 1, and the following identities hold%

\begin{equation}
\dot{\varphi}=J_{\varphi}v,\dot{\theta}=J_{\theta}v,\ \ \ \ J_{\varphi}%
^{2}+\sin^{2}(\varphi)J_{\theta}^{2}=1 \label{3.4}%
\end{equation}

Next, consider the two positively oriented orthonormal moving frames
$(\mathbf{\tau}^{\ast},\mathbf{\nu}^{\ast})$ and $(\frac{\partial}%
{\partial\varphi},\frac{1}{\sin\varphi}\frac{\partial}{\partial\theta})$ along
the oriented spherical curve $\gamma^{\ast}$, where $\mathbf{\tau}^{\ast}$
(resp. $\mathbf{\nu}^{\ast}$) is the unit tangent (resp. normal) vector.
Writing the frames formally as column vectors we can express their
relationship at each point $p$ by a rotation matrix defined by the direction
element, namely
\begin{equation}
(\mathbf{\tau}^{\ast},\mathbf{\nu}^{\ast})^{T}=\left(
\begin{array}
[c]{cc}%
J_{\varphi} & J_{\theta}\sin\varphi\\
-J_{\theta}\sin\varphi & J_{\varphi}%
\end{array}
\right)  (\frac{\partial}{\partial\varphi},\frac{1}{\sin\varphi}\frac
{\partial}{\partial\theta})^{T}, \label{frames}%
\end{equation}

On the other hand, the gradient field of $U^{\ast}$ on the sphere is
\[
\nabla U^{\ast}=U_{\varphi}^{\ast}\frac{\partial}{\partial\varphi}%
+\frac{U_{\theta}^{\ast}}{\sin^{2}\varphi}\frac{\partial}{\partial\theta},
\]
and its inner product with $\mathbf{\tau}^{\ast}$ and $\mathbf{\nu}^{\ast}$
yields the \emph{tangential }and \emph{normal derivative} of $U^{\ast}$ along
the curve, namely the pair $(U_{\mathbf{\tau}}^{\ast},U_{\mathbf{\nu}}^{\ast
}).$ The latter is, in fact, related to the pair of partial derivatives
$(U_{\varphi}^{\ast},\frac{1}{\sin\varphi}U_{\theta}^{\ast})$ via the same
matrix as in (\ref{frames}), as follows
\begin{equation}
(U_{\mathbf{\tau}}^{\ast},U_{\mathbf{\nu}}^{\ast})^{T}=\left(
\begin{array}
[c]{cc}%
J_{\varphi} & J_{\theta}\sin\varphi\\
-J_{\theta}\sin\varphi & J_{\varphi}%
\end{array}
\right)  (U_{\varphi}^{\ast},\frac{1}{\sin\varphi}U_{\theta}^{\ast})^{T}
\label{U-deriv}%
\end{equation}
In particular, the left side is an intrinsic invariant of order 1 which
represents the gradient field $\nabla U^{\ast}$ along the curve $\gamma^{\ast
}$.

\subsubsection{Geodesic curvature of the shape curve}

The geometry of the spherical curve $\gamma^{\ast}$ itself is encoded into its
geodesic curvature function $K^{\ast}(s)$, and in general $\gamma^{\ast}(s)$
is in fact completely determined by the intrinsic function $K^{\ast}$ and the
initial direction of $\gamma^{\ast}$. Below we shall calculate and deduce an
expression for $K^{\ast}(s)$ in terms of simpler invariants of order $\leq1$.

One way to calculate $K^{\ast}$ is to express $\gamma^{\ast}$ in Euclidean
coordinates as $\mathbf{x}(s)=(x(s),y(s),z(s))$ and use the formula%
\[
K^{\ast}(s)=\mathbf{x}(s)\times\mathbf{x}^{\prime}(s)\cdot\mathbf{x}%
^{\prime\prime}(s)
\]
where differentiation is with respect to arc-length $s$. Then, by returning to
spherical coordinates and writing $\varphi^{\prime}=d\varphi/ds$ etc.
\begin{align}
K^{\ast}  &  =(\cos\varphi)\theta^{\prime}(1+\varphi^{\prime2})+\sin
\varphi(\varphi^{\prime}\theta^{\prime\prime}-\theta^{\prime}\varphi
^{\prime\prime})\nonumber\\
&  =\frac{1}{v^{3}}\left\{  (\cos\varphi)\dot{\theta}(v^{2}+\dot{\varphi}%
^{2})+\sin\varphi(\dot{\varphi}\ddot{\theta}-\dot{\theta}\ddot{\varphi
})\right\}  \text{,} \label{K*(t)}%
\end{align}
where in the first line the expression for $K^{\ast}$ is intrinsic and the
second is an expression in temporal invariants.

Next, let us eliminate the second order terms $\ddot{\varphi}$ and
$\ddot{\theta}$ in the expression (\ref{K*(t)}), using equations (ii), (iii)
of the ODE system (\ref{NewtonRedu}). The ensuing calculations
\begin{align*}
K^{\ast}v^{3}  &  =(\cos\varphi)\dot{\theta}(v^{2}+\dot{\varphi}^{2}%
)+(\sin\varphi)\dot{\varphi}\left(  -\frac{2\dot{\rho}}{\rho}\dot{\theta
}-2(\cot\varphi)\dot{\varphi}\dot{\theta}-\frac{2\omega}{\rho^{2}\sin\varphi
}\dot{\varphi}+\frac{4}{\rho^{3}}\frac{1}{\sin^{2}\varphi}U_{\theta}^{\ast
}\right) \\
&  -(\sin\varphi)\dot{\theta}\left(  -\frac{2\dot{\rho}}{\rho}\dot{\varphi
}+\frac{1}{2}\sin(2\varphi)\dot{\theta}^{2}+\frac{2\omega\sin\varphi}{\rho
^{2}}\dot{\theta}+\frac{4}{\rho^{3}}U_{\varphi}^{\ast}\right) \\
&  =\frac{4}{\rho^{3}}\left(  \frac{\dot{\varphi}}{\sin\varphi}U_{\theta
}^{\ast}-\dot{\theta}\sin\varphi U_{\varphi}^{\ast}\right)  -\frac{2\omega
}{\rho^{2}}(\dot{\varphi}^{2}+\sin^{2}\varphi\dot{\theta}^{2})=\frac{4v}%
{\rho^{3}}U_{\mathbf{\nu}}^{\ast}-\frac{2\omega}{\rho^{2}}v^{2}%
\end{align*}
lead us to the fundamental curvature formula\emph{ }$(\ref{K*})$, which we may
also write as
\begin{equation}
\rho^{3}v^{2}+\frac{2\omega}{K^{\ast}}\rho v=4\mathfrak{S} \label{K*b}%
\end{equation}
where
\begin{equation}
\mathfrak{S}=\frac{U_{\nu}^{\ast}}{K^{\ast}} \label{Siegel}%
\end{equation}
is the intrinsic \emph{Siegel function} introduced in \cite{HS-2}. This
function neatly encodes the relative geometry of the pair $(\mathfrak{\gamma
}^{\ast},\nabla U^{\ast})$, and it is also independent of the orientation of
$\gamma^{\ast}$.

In the special case that $K^{\ast}$ vanishes, the above calculation of
$K^{\ast}$ yields the identity%
\begin{equation}
\rho v=\frac{2U_{\nu}^{\ast}}{\omega} \label{K*c}%
\end{equation}

\subsubsection{Power series expansions of functions on the shape sphere}

Of primary interest are the solution curves $\bar{\gamma}(t)$ of the ODE
(\ref{NewtonRedu}), whose coordinate functions $\rho(t),\varphi(t),\theta(t)$
are temporal invariants, and by definition, so are the coefficients of their
series expansions (\ref{intcurve}). Furthermore, evaluation of functions on
the shape sphere along the curve $\gamma^{\ast}(t)$ $=(\varphi(t),\theta(t))$
yields series expansions in time $t$, such as%
\begin{align}
U^{\ast}  &  =u_{0}+\grave{u}_{1}t+\grave{u}_{2}t^{2}+...\nonumber\\
U_{\varphi}^{\ast}  &  =\mu_{0}+\mu_{1}t+\mu_{2}t^{2}+...\label{3.17}\\
U_{\theta}^{\ast}  &  =\eta_{0}+\eta_{1}t+\eta_{2}t^{2}+....\nonumber
\end{align}
The coefficients are temporal invariants, except the leading terms which are,
by definition, intrinsic and of order zero.

Moreover, we also need to expand the shape speed
\[
v=v_{0}+v_{1}t+v_{2}t^{2}+...
\]
whose first two coefficients are readily obtained from (\ref{3.2} )
\begin{align}
v_{0}  &  =\sqrt{\varphi_{1}^{2}+g_{0}\theta_{1}^{2}}\text{, \ }\label{v0}\\
\text{\ }v_{1}  &  =\frac{1}{v_{0}}[2\varphi_{1}\varphi_{2}+\frac{f_{0}}%
{2}\varphi_{1}\theta_{1}^{2}+2g_{0}\theta_{1}\theta_{2}], \label{v1}%
\end{align}
where we have simplified notation by setting
\[
f_{0}=\sin2\varphi_{0}\text{, \ }g_{0}=\sin^{2}\varphi_{0}%
\]
Clearly, $v_{0}$ and $v_{1}$ are invariants of order $1$ and $2$, respectively.

The change of direction of a shape curve is expressed by functions such as%

\[
J_{\varphi}^{\prime}=\frac{d}{ds}J_{\varphi},\text{ }J_{\theta}^{\prime}%
=\frac{d}{ds}J_{\theta},\text{ }J_{\varphi}^{\prime\prime}=\frac{d^{2}}%
{ds^{2}}J_{\varphi},\text{ etc.}\
\]
and we shall use the same notation for their evaluation at $s=0$, for example
by (\ref{3.4})
\ \ \ \ \ \ \ \ \ \ \ \ \ \ \ \ \ \ \ \ \ \ \ \ \ \ \ \ \ \ \ \ \ \ \ \ \ \ \ \ \ \ \ \ \ \ \ \ \ \ \ \ \ \ \ \ \ \ \ \ \ \
\begin{equation}
J_{\varphi}^{\prime}=\frac{1}{v_{0}^{2}}(2\varphi_{2}-v_{1}J_{\varphi}),\text{
}J_{\theta}^{\prime}=\frac{1}{v_{0}^{2}}(2\theta_{2}-v_{1}J_{\theta})\text{,
etc. } \label{Jprime1}%
\end{equation}
These are intrinsic invariants of order 2, although the above expressions
involve temporal ones $\varphi_{2},\theta_{2},v_{1}$ of order 2.

Next, the coefficients of series expansions with respect to arc-length $s$ of
$\gamma^{\ast}$, such as
\begin{align}
U^{\ast}  &  =u_{0}+u_{1}s+u_{2}s^{2}+..\nonumber\\
U_{\tau}^{\ast}  &  =\tau_{0}+\tau_{1}s+\tau_{2}s^{2}+...(note:\tau
_{k-1}=ku_{k})\nonumber\\
U_{\nu}^{\ast}  &  =w_{0}+w_{1}s+w_{2}s^{2}+...\label{3.16}\\
K^{\ast}  &  =K_{0}+K_{1}s+K_{2}s^{2}+...\nonumber\\
\mathfrak{S}  &  =\mathfrak{S}_{0}+\mathfrak{S}_{1}s+\mathfrak{S}_{2}%
s^{2}+...\nonumber
\end{align}
are intrinsic invariants. The beginning coefficients are calculated using
(\ref{frames})-(\ref{U-deriv}) and (\ref{K*(t)}):
\begin{align}
u_{0}  &  =U^{\ast}(\varphi_{0},\theta_{0}),\text{ \ }\tau_{0}=u_{1}%
=U_{\varphi}^{\ast}J_{\varphi}+U_{\theta}^{\ast}J_{\theta}\ \nonumber\\
\tau_{1}  &  =U_{\varphi}^{\ast}J_{\varphi}^{\prime}+U_{\theta}^{\ast
}J_{\theta}^{\prime}+(U_{\varphi\varphi}^{\ast}J_{\varphi}^{2}+2U_{\varphi
\theta}^{\ast}J_{\varphi}J_{\theta}+U_{\theta\theta}^{\ast}J_{\theta}%
^{2})\nonumber\\
\text{ }w_{0}  &  =(\frac{U_{\theta}^{\ast}}{g_{0}}J_{\varphi}-U_{\varphi
}^{\ast}J_{\theta})\sin\varphi_{0}\label{first coeff}\\
w_{1}  &  =\frac{U_{\theta}^{\ast}}{\sin\varphi_{0}}J_{\varphi}^{\prime
}-U_{\varphi}^{\ast}\sin\varphi_{0}J_{\theta}^{\prime}-\frac{U_{\theta}^{\ast
}\cos\varphi_{0}}{g_{0}}J_{\varphi}^{2}+(\frac{U_{\theta\theta}^{\ast}}%
{\sin\varphi_{0}}-U_{\varphi\theta}^{\ast}\sin\varphi_{0})J_{\theta}%
^{2}\nonumber\\
&  +(\frac{U_{\theta\theta}^{\ast}}{\sin\varphi_{0}}-U_{\varphi\varphi}^{\ast
}\sin\varphi_{0}-U_{\varphi}^{\ast}\cos\varphi_{0})J_{\varphi}J_{\theta
}\nonumber\\
K_{0}  &  =J_{\theta}(1+J_{\varphi}^{2})\cos\varphi_{0}+(J_{\varphi}J_{\theta
}^{\prime}-J_{\theta}J_{\varphi}^{\prime})\sin\varphi_{0}\nonumber\\
K_{1}  &  =[-J_{\varphi}J_{\theta}(1+J_{\varphi}^{2})+J_{\varphi}J_{\theta
}^{\prime\prime}-J_{\theta}J_{\varphi}^{\prime\prime}]\sin\varphi
_{0}+[J_{\theta}^{\prime}(1+2J_{\varphi}^{2})+J_{\varphi}J_{\theta}J_{\varphi
}^{\prime}\ ]\cos\varphi_{0}\nonumber\\
\mathfrak{S}_{0}  &  =\frac{w_{0}}{K_{0}},\text{ }\mathfrak{S}_{1}=\frac
{K_{0}w_{1}-K_{1}w_{0}}{K_{0}^{2}}\nonumber
\end{align}
Clearly the order of a coefficient $a_{k}$ is $k$ plus the order of $a_{0}$,
and we note that $u_{0}$ has order $0$, $\tau_{0},v_{0}$ and $w_{0}$ have
order $1$, but $K_{0}$ has order $2$.

\subsubsection{Singularities of the shape curve}

A point $p=(\varphi_{0},\theta_{0})$ on the spherical curve $\gamma^{\ast}$ is
said to be \emph{regular} if $K^{\ast}\neq0$ and $U_{\nu}^{\ast}\neq0$ at $p$,
otherwise it is called \emph{singular, }and a curve with no regular point is
called\emph{ exceptional. }Note the geometric meaning of $U_{\nu}^{\ast}=0$,
namely the curve $\gamma^{\ast}$ is tangential to the gradient flow of
$U^{\ast}$ at $p$.

Now consider the time parametrization of $\gamma^{\ast}$. A point
$p=\gamma^{\ast}(t_{0})$ is a \emph{cusp }of the curve $\gamma^{\ast}(t)$ if
the speed $v$ vanishes at $t=t_{0}$. Assuming (for simplicity) that $K^{\ast}%
$is defined (or bounded) at $p$, it follow from (\ref{K*b})%
\[
\lim_{t\rightarrow t_{0}}\frac{U_{\nu}^{\ast}}{v}=\frac{1}{2}\omega\rho_{0},
\]
so $U_{\nu}^{\ast}=0$ but possibly $K^{\ast}\neq0$ at $p$. In particular, $p$
is a singular point. Conversely, by (\ref{K*b}) and assuming $\omega\neq0$, a
point where both $U_{\nu}^{\ast}$ and $K^{\ast}$ vanish must be a cusp. This
does not hold for $\omega=0$; for example, the shape curve of the figure eight
periodic motion (cf. \cite{CM}) passes through the Euler points on the equator
circle with $v\neq0$ and $K^{\ast}=$ $U_{\nu}^{\ast}=0$. $\ $

Examples of 3-body motions with exceptional shape curve arise from the
isosceles solutions of the 3-body problem, which are fairly well understood
(cf. e.g. \cite{Cabral}). The isosceles m-triangle has two equal masses at the
base, and the shapes $\gamma^{\ast}$ of these m-triangles constitute a
longitude circle on the shape sphere $M^{\ast}$, and hence $K^{\ast}$ vanishes
along the curve. In fact, $U_{\mathbf{\nu}}^{\ast}$ also vanishes on this
circle since it is a gradient line for $U^{\ast}$. We point out, however, an
isosceles triangle motion in the plane must have $\omega=0$, so in our study
these motions are excluded at the outset.

The shape curves of collinear motions, being confined to the equatorial circle
of $M^{\ast}$, satisfy $K^{\ast}=U_{\mathbf{\nu}}^{\ast}=0$, so they must be
regarded as exceptional. It is not difficult to see (purely kinematically)
that a collinear motion is planar, and is even confined to a fixed line if
$\omega=0$.

On the other hand, takining a closer look at the Newtonian case with
$\omega\neq0$, we observe that a\ collinear motion must have constant shape
and hence the shape curve is a single point, namely one of the three Euler
points. In fact, from the identity (\ref{K*b}) it follows that $v=0$ at each
point, so $\gamma^{\ast}$ must be a single point. The vanishing of $v$ also
follows directly from the ODE (\ref{NewtonRedu}), where in equation (ii) we
have, by assumption, $\varphi=\pi/2,\dot{\varphi}=0,U_{\varphi}^{\ast}=0$,
hence also $v=\dot{\theta}=0$. Then, from the ODE it follows that $U_{\varphi
}^{\ast}$ $=U_{\theta}^{\ast}=0$, namely the point is critical for $U^{\ast}$
and hence, by definition, is an Euler point.

\section{ A three-body motion is essentially determined by its geometric shape
curve}

The purpose of this section is to interpret properly and provide evidence for
the following:

\begin{conjecture}
The geometric shape curve $\gamma^{\ast}$ of a planar 3-body motion determines
the time parametrized moduli curve $\bar{\gamma}(t)$, and hence determines the
3-body motion modulo a fixed rotation of the plane.
\end{conjecture}

This should hold in general, with some "obvious" exceptions. So, we are aiming
at an analytical reconstruction of the moduli curve $\bar{\gamma}(t)$, based
on purely geometric data concerning the shape curve $\gamma^{\ast}$ and its
interaction with the gradient field $\nabla U^{\ast}$. This is summarized as follows:

\begin{theorem}
Assume an oriented geometric curve $\gamma^{\ast}$ on the 2-sphere is
realizable as the shape curve of a planar 3-body motion $\gamma(t)$ at a given
energy-momentum level $(h,\omega)$. Then the relative geometry of
$(\gamma^{\ast},\nabla U^{\ast})$ on the shape sphere in the neighborhood of a
generic point $p$ of $\gamma^{\ast}$ yields the information needed to
determine the moduli curve $\bar{\gamma}(t)$ as a solution of the ODE
(\ref{NewtonRedu}).
\end{theorem}

The proof will be elaborated in the following subsections, through local
analysis which eventually reduces the problem to solving an algebraic system
(\ref{xyz}) involving three equations and three variables. This system depends
only on intrinsic invariants, and it is the solution of this system, with some
ambiguity, which enables us to determine the initial data (\ref{initial})
which generate the curve $\bar{\gamma}(t)$ as a solution of the system
(\ref{NewtonRedu}).

\subsection{Explicit calculation of the moduli curve as an initial value
problem}

Since the ODE system (\ref{NewtonRedu}) is analytic, one can develop
recursively the power series expansion of the solution $\bar{\gamma}%
(t)=(\rho(t),\varphi(t),\theta(t))$ by the method of undetermined
coefficients. Namely, starting from the initial data set (\ref{initial}), the
temporal invariants $\rho_{k},\varphi_{k},\theta_{k}$ of order $k\geq2$ are
calculated recursively and expressed in terms of temporal invariants of lower order.

More specifically, repeated differentiation of the system (\ref{NewtonRedu})
yields identities of type
\begin{align}
E_{1,n}  &  :0=(n+2)(n+1)\rho_{0}^{2}\rho_{n+2}+.....\nonumber\\
E_{2,n}  &  :0=(n+2)(n+1)\rho_{0}^{3}\varphi_{n+2}+.....\label{En}\\
E_{3,n}  &  :0=(n+2)(n+1)g_{0}\rho_{0}^{3}\theta_{n+2}+....\nonumber
\end{align}
Thus, the identity $E_{i,n}$ expresses a unique temporal invariant of highest
order $n+2$ in terms of lower order invariants. This procedure starts with the
case $n=0$ which yields the following identities stated for convenience
\begin{align}
E_{10}  &  :0=2\rho_{0}^{2}\rho_{2}+\rho_{0}\rho_{1}^{2}-2h\rho_{0}%
-u_{0}\nonumber\\
E_{20}  &  :0=2\rho_{0}^{3}\varphi_{2}+2\rho_{0}^{2}\rho_{1}\varphi
_{1}-2\omega(\sin\varphi_{0})\rho_{0}\theta_{1}-\frac{1}{2}f_{0}\rho_{0}%
^{3}\theta_{1}^{2}-4\mu_{0}\label{E0}\\
E_{30}  &  :0=2g_{0}\rho_{0}^{3}\theta_{2}+2g_{0}\rho_{0}^{2}\rho_{1}%
\theta_{1}+2\omega(\sin\varphi_{0})\rho_{0}\varphi_{1}+f_{0}\rho_{0}%
^{3}\varphi_{1}\theta_{1}-4\eta_{0}\nonumber
\end{align}

However, we want to determine the initial data (\ref{initial}) solely in terms
of intrinsic invariants and in the simplest way. A natural first step in this
direction is the replacement of the six-tuple (\ref{initial}) by another
equivalent "six-tuple"
\begin{equation}
\lbrack\rho_{0},\varphi_{0},\theta_{0};\rho_{1},\varphi_{1},\theta
_{1}]\longleftrightarrow\lbrack\rho_{0},\rho_{1},v_{0};(J_{\varphi},J_{\theta
})_{p}], \label{1st reduction}%
\end{equation}
namely the initial data (\ref{initial}) is replaced by the temporal invariants
$\rho_{0},\rho_{1},v_{0}$ and the direction element of $\gamma^{\ast}$ at
$p=(\varphi_{0},\theta_{0})$, the latter being an intrinsic invariant. This is
so because the pair $(\varphi_{1},\theta_{1})$, representing the velocity of
the shape curve $\gamma^{\ast}(t)$ at $p$, is determined by the initial speed
$v_{0}$ and direction $(J_{\varphi},J_{\theta})$ at $p$. The triple $(\rho
_{0},\rho_{1},v_{0})$ shall be referred to as the \emph{basic temporal
invariants.}

For convenience, we state the following preliminary observation:

\begin{proposition}
\label{BasicProp} For a given direction element $(J_{\varphi},J_{\theta})_{p}$
and basic temporal invariants $(\rho_{0},\rho_{1},v_{0})$, we can determine
the initial data set (\ref{initial}) and hence calculate successively, using
the equations (\ref{En}), the power series expansion of the time parametrized
moduli curve $\bar{\gamma}(t)$. In other words, the direction element and the
triple $(\rho_{0},\rho_{1},v_{0})$ determine a unique solution $\bar{\gamma
}(t)$ of the ODE system (\ref{NewtonRedu}).
\end{proposition}

\subsubsection{Reduction of order}

The\ notion of order of an invariant, as defined in Section 2.4.1, makes
little sense when the coordinate functions $\rho(t),\varphi(t),\theta(t)$
solve the ODE system (\ref{NewtonRedu}). Indeed, by means of this system a
differential invariant $Q$ of order 2 reduces its order to 1 by substituting
the expressions for $\ddot{\rho},\ddot{\varphi},\ddot{\theta}$ into $Q$.
Furthermore, expressions for higher order derivatives of $\rho,\varphi,\theta$
are found by successive differentiation of the differential equations, and so
they can be substituted into higher order invariants to reduce their total
order. Thus, any given differential invariant $Q$ of order $n$ can be reduced
stepwise to a functional expression of order $\leq1$. Then by evaluation at
$t=0$ the resulting terms involve only temporal invariants of order $\leq1$,
possibly also intrinsic invariants of order $0$, that is, the function
$U^{\ast}$ or its partial derivatives (of any "order") evaluated at $p.$

A \underline{complete} \underline{reduction} of an invariant is achieved when
all invariants in the reduced expression have order $\leq1$. Let us work out
complete reductions of the basic curvature invariants $K_{0},K_{1},w_{1}$
listed in (\ref{first coeff}). First of all, recall that the expression for
$K_{0}$ in (\ref{first coeff}) follows directly from its definition, as an
intrinsic invariant of order 2. However, its complete reduction follows
immediately from (\ref{K*}) or (\ref{K*b}), so the following two expressions
for $K_{0}$ must be equal
\begin{equation}
K_{0}=\frac{-2}{\rho_{0}^{2}v_{0}}(\omega-\frac{2w_{0}}{\rho_{0}v_{0}%
})=J_{\theta}(1+J_{\varphi}^{2})\cos\varphi_{0}+(J_{\varphi}J_{\theta}%
^{\prime}-J_{\theta}J_{\varphi}^{\prime})\sin\varphi_{0}, \label{K0}%
\end{equation}
which can also be verified directly by reducing the 2nd order invariants
$J_{\varphi}^{\prime}$ and $J_{\theta}^{\prime}$, see below.

\begin{lemma}
The 2nd order temporal invariant $v_{1}$ has the following complete
reduction\ \ \ \ \ \ \ \ \ \ \ \ \ \ \ \ \ \ \ \ \ \ \ \ \ \ \ \ \ \ \ \ \ \ \ \ \ \ \ \ \ \ \ \ \ \ \ \ \ \ \ \ \ \ \ \ \ \ \ \ \ \ \ \ \ \ \ \ \ \ \ \ \ \ \ \ \ \ \ \ \ \ \ \ \ \ \ \ \ \ \ \ \ \ \ \ \ \ \ \ \ \ \ \ \ \ \ \ \ \ \ \ \ \ \ \ \ \ \ \ \ \ \ \ \ \ \ \ \ \ \ \ \ \ \ \ \ \ \ \ \ \ \ \ \ \ \ \ \ \ \ \ \ \ \ \ \ \ \ \ \ \ \ \ \ \ \ \ \ \ \ \ \ \ \ \ \ \ \ \ \ \ \ \ \ \ \ \ \ \ \ \ \ \ \ \ \ \ \ \ \ \ \ \ \ \ \ \ \ \ \ \ \ \ \ \ \ \ \ \ \ \ \ \ \ \ \ \ \ \ \ \ \ \ \ \ \ \ \ \ \ \ \ \ \ \ \ \ \ \ \ \ \ \ \ \ \ \ \ \ \ \ \ \ \ \ \ \ \ \ \ \ \ \ \ \ \ \ \ \ \ \ \ \ \ \ \ \ \ \ \ \ \ \ \ \ \ \ \ \ \ \ \ \ \ \ \ \ \ \ \ \ \ \ \ \ \ \ \ \ \ \ \ \ \ \ \ \ \ \ \ \ \ \ \ \ \ \ \ \ \ \ \ \ \ \ \ \ \ \ \ \ \ \ \ \ \ \ \ \ \ \ \ \ \ \ \ \ \ \ \ \ \ \ \ \ \ \ \ \ \ \ \ \ \ \ \ \ \ \ \ \ \ \ \ \ \ \ \ \ \ \ \ \ \ \ \ \ \ \ \ \ \ \ \ \ \ \ \ \ \ \ \ \ \ \ \ \ \ \ \ \ \ \ \ \ \ \ \ \ \ \ \ \ \ \ \ \ \ \ \ \ \ \ \ \ \ \ \ \ \ \ \ \ \ \ \ \ \ \ \ \ \ \ \ \ \ \ \ \ \ \ \ \ \ \ \ \ \ \ \ \ \ \ \ \ \ \ \ \ \ \ \ \ \ \ \ \ \ \ \ \ \ \ \ \ \ \ \ \ \ \ \ \ \ \ \ \ \ \ \ \ \ \ \ \ \ \ \ \ \ \ \ \ \ \ \ \ \ \ \ \ \ \ \ \ \ \ \ \ \ \ \ \ \ \ \ \ \ \ \ \ \ \ \ \ \ \ \ \ \ \ \ \ \ \ \ \ \ \ \ \ \ \ \ \ \ \ \ \ \ \ \ \ \ \ \ \ \ \ \ \ \ \ \ \ \ \ \ \ \ \ \ \ \ \ \ \ \ \ \ \ \ \ \ \ \ \ \ \ \ \ \ \ \ \ \ \ \ \ \ \ \ \ \ \ \ \ \ \ \ \ \ \ \ \ \ \ \ \ \ \ \ \ \ \ \ \ \ \ \ \ \ \ \ \ \ \ \ \ \ \ \ \ \ \ \ \ \ \ \ \ \ \ \ \ \ \ \ \ \ \ \ \ \ \ \ \ \ \ \ \ \ \ \ \ \ \ \ \ \ \ \ \ \ \ \ \ \ \ \ \ \ \ \ \ \ \ \ \ \ \ \ \ \ \ \ \ \ \ \ \ \ \ \ \ \ \ \ \ \ \ \ \ \ \ \ \ \ \ \ \ \ \ \ \ \ \ \ \ \ \ \ \ \ \ \ \ \ \ \ \ \ \ \ \ \ \ \ \ \ \ \ \ \ \ \ \ \ \ \ \ \ \ \ \ \ \ \ \ \ \ \ \ \ \ \ \ \ \ \ \ \ \ \ \ \ \ \ \ \ \ \ \ \ \ \ \ \ \ \ \ \ \ \ \ \ \ \ \ \ \ \ \ \ \ \ \ \ \ \ \ \ \ \ \ \ \ \ \ \ \ \ \ \ \ \ \ \ \ \ \ \ \ \ \ \ \ \ \ \ \ \ \ \ \ \ \ \ \ \ \ \ \ \ \ \ \ \ \ \ \ \ \ \ \ \ \ \ \ \ \ \ \ \ \ \ \ \ \ \ \ \ \ \ \ \ \ \ \ \ \ \ \ \ \ \ \ \ \ \ \ \ \ \ \ \ \ \ \ \ \ \ \ \ \ \ \ \ \ \ \ \ \ \ \ \ \ \ \ \ \ \ \ \ \ \ \ \ \ \ \ \ \ \ \ \ \ \ \ \ \ \ \ \ \ \ \ \ \ \ \ \ \ \ \ \ \ \ \ \ \ \ \ \ \ \ \ \ \ \ \ \ \ \ \ \ \ \ \ \ \ \ \ \ \ \ \ \ \ \ \ \ \ \ \ \ \ \ \ \ \ \ \ \ \ \ \ \ \ \ \ \ \ \ \ \ \ \ \ \ \ \ \ \ \ \ \ \ \ \ \ \ \ \ \ \ \ \ \ \ \ \ \ \ \ \ \ \ \ \ \ \ \ \ \ \ \ \ \ \ \ \ \ \ \ \ \ \ \ \ \ \ \ \ \ \ \ \ \ \ \ \ \ \ \ \ \ \ \ \ \ \ \ \ \ \ \ \ \ \ \ \ \ \ \ \ \ \ \ \ \ \ \ \ \ \ \ \ \ \ \ \ \ \ \ \ \ \ \ \ \ \ \ \ \ \ \ \ \ \ \ \ \ \ \ \ \ \ \ \ \ \ \ \ \ \ \ \ \ \ \ \ \ \ \ \ \ \ \ \ \ \ \ \ \ \ \ \ \ \ \ \ \ \ \ \ \ \ \ \ \ \ \ \ \ \ \ \ \ \ \ \ \ \ \ \ \ \ \ \ \ \ \ \ \ \ \ \ \ \ \ \ \ \ \ \ \ \ \ \ \ \ \ \ \ \ \ \ \ \ \ \ \ \ \ \ \ \ \ \ \ \ \ \ \ \ \ \ \ \ \ \ \ \ \ \ \ \ \ \ \ \ \ \ \ \ \ \ \ \ \ \ \ \ \ \ \ \ \ \ \ \ \ \ \ \ \ \ \ \ \ \ \ \ \ \ \ \ \ \ \ \ \ \ \ \ \ \ \ \ \ \ \ \ \ \ \ \ \ \ \ \ \ \ \ \ \ \ \ \ \ \ \ \ \ \ \ \ \ \ \ \ \ \ \ \ \ \ \ \ \ \ \ \ \ \ \ \ \ \ \ \ \ \ \ \ \ \ \ \ \ \ \ \ \ \ \ \ \ \ \ \ \ \ \ \ \ \ \ \ \ \ \ \ \ \ \ \ \ \ \ \ \ \ \ \ \ \ \ \ \ \ \ \ \ \ \ \ \ \ \ \ \ \ \ \ \ \ \ \ \ \ \ \ \ \ \ \ \ \ \ \ \ \ \ \ \ \ \ \ \ \ \ \ \ \ \ \ \ \ \ \ \ \ \ \ \ \ \ \ \ \ \ \ \ \ \ \ \ \ \ \ \ \ \ \ \ \ \ \ \ \ \ \ \ \ \ \ \ \ \ \ \ \ \ \ \ \ \ \ \ \ \ \ \ \ \ \ \ \ \ \ \ \ \ \ \ \ \ \ \ \ \ \ \ \ \ \ \ \ \ \ \ \ \ \ \ \ \ \ \ \ \ \ \ \ \ \ \ \ \ \ \ \ \ \ \ \ \ \ \ \ \ \ \ \ \ \ \ \ \ \ \ \ \ \ \ \ \ \ \ \ \ \ \ \ \ \ \ \ \ \ \ \ \ \ \ \ \ \ \ \ \ \ \ \ \ \ \ \ \ \ \ \ \ \ \ \ \ \ \ \ \ \ \ \ \ \ \ \ \ \ \ \ \ \ \ \ \ \ \ \ \ \ \ \ \ \ \ \ \ \ \ \ \ \ \ \ \ \ \ \ \ \ \ \ \ \ \ \ \ \ \ \ \ \ \ \ \ \ \ \ \ \ \ \ \ \ \ \ \ \ \ \ \ \ \ \ \ \ \ \ \ \ \ \ \ \ \ \ \ \ \ \ \ \ \ \ \ \ \ \ \ \ \ \ \ \ \ \ \ \ \ \ \ \ \ \ \ \ \ \ \ \ \ \ \ \ \ \ \ \ \ \ \ \ \ \ \ \ \ \ \ \ \ \ \ \ \ \ \ \ \ \ \ \ \ \ \ \ \ \ \ \ \ \ \ \ \ \ \ \ \ \ \ \ \ \ \ \ \ \ \ \ \ \ \ \ \ \ \ \ \ \ \ \ \ \ \ \ \ \ \ \ \ \ \ \ \ \ \ \ \ \ \ \ \ \ \ \ \ \ \ \ \ \ \ \ \ \ \ \ \ \ \ \ \ \ \ \ \ \ \ \ \ \ \ \ \ \ \ \ \ \ \ \ \ \ \ \ \ \ \ \ \ \ \ \ \ \ \ \ \ \ \ \ \ \ \ \ \ \ \ \ \ \ \ \ \ \ \ \ \ \ \ \ \ \ \ \ \ \ \ \ \ \ \ \ \ \ \ \ \ \ \ \ \ \ \ \ \ \ \ \ \ \ \ \ \ \ \ \ \ \ \ \ \ \ \ \ \ \ \ \ \ \ \ \ \ \ \ \ \ \ \ \ \ \ \ \ \ \ \ \ \ \ \ \ \ \ \ \ \ \ \ \ \ \ \ \ \ \ \ \ \ \ \ \ \ \ \ \ \ \ \ \ \ \ \ \ \ \ \ \ \ \ \ \ \ \ \ \ \ \ \ \ \ \ \ \ \ \ \ \ \ \ \ \ \ \ \ \ \ \ \ \ \ \ \ \ \ \ \ \ \ \ \ \ \ \ \ \ \ \ \ \ \ \ \ \ \ \ \ \ \ \ \ \ \ \ \ \ \ \ \ \ \ \ \ \ \ \ \ \ \ \ \ \ \ \ \ \ \ \ \ \ \ \ \ \ \ \ \ \ \ \ \ \ \ \ \ \ \ \ \ \ \ \ \ \ \ \ \ \ \ \ \ \ \ \ \ \ \ \ \ \ \ \ \ \ \ \ \ \ \ \ \ \ \ \ \ \ \ \ \ \ \ \ \ \ \ \ \ \ \ \ \ \ \ \ \ \ \ \ \ \ \ \ \ \ \ \ \ \ \ \ \ \ \ \ \ \ \ \ \ \ \ \ \ \ \ \ \ \ \ \ \ \ \ \ \ \ \ \ \ \ \ \ \ \ \ \ \ \ \ \ \ \ \ \ \ \ \ \ \ \ \ \ \ \ \ \ \ \ \ \ \ \ \ \ \ \ \ \ \ \ \ \ \ \ \ \ \ \ \ \ \ \ \ \ \ \ \ \ \ \ \ \ \ \ \ \ \ \ \ \ \ \ \ \ \ \ \ \ \ \ \ \ \ \ \ \ \ \ \ \ \ \ \ \ \ \ \ \ \ \ \ \ \ \ \ \ \ \ \ \ \ \ \ \ \ \ \ \ \ \ \ \ \ \ \ \ \ \ \ \ \ \ \ \ \ \ \ \ \ \ \ \ \ \ \ \ \ \ \ \ \ \ \ \ \ \ \ \ \ \ \ \ \ \ \ \ \ \ \ \ \ \ \ \ \ \ \ \ \ \ \ \ \ \ \ \ \ \ \ \ \ \ \ \ \ \ \ \ \ \ \ \ \ \ \ \ \ \ \ \ \ \ \ \ \ \ \ \ \ \ \ \ \ \ \ \ \ \ \ \ \ \ \ \ \ \ \ \ \ \ \ \ \ \ \ \ \ \ \ \ \ \ \ \ \ \ \ \ \ \ \ \ \ \ \ \ \ \ \ \ \ \ \ \ \ \ \ \ \ \ \ \ \ \ \ \ \ \ \ \ \ \ \ \ \ \ \ \ \ \ \ \ \ \ \ \ \ \ \ \ \ \ \ \ \ \ \ \ \ \ \ \ \ \ \ \ \ \ \ \ \ \ \ \ \ \ \ \ \ \ \ \ \ \ \ \ \ \ \ \ \ \ \ \ \ \ \ \ \ \ \ \ \ \ \ \ \ \ \ \ \ \ \ \ \ \ \ \ \ \ \ \ \ \ \ \ \ \ \ \ \ \ \ \ \ \ \ \ \ \ \ \ \ \ \ \ \ \ \ \ \ \ \ \ \ \ \ \ \ \ \ \ \ \ \ \ \ \ \ \ \ \ \ \ \ \ \ \ \ \ \ \ \ \ \ \ \ \ \ \ \ \ \ \ \ \ \ \ \ \ \ \ \ \ \ \ \ \ \ \ \ \ \ \ \ \ \ \ \ \ \ \ \ \ \ \ \ \ \ \ \ \ \ \ \ \ \ \ \ \ \ \ \ \ \ \ \ \ \ \ \ \ \ \ \ \ \ \ \ \ \ \ \ \ \ \ \ \ \ \ \ \ \ \ \ \ \ \ \ \ \ \ \ \ \ \ \ \ \ \ \ \ \ \ \ \ \ \ \ \ \ \ \ \ \ \ \ \ \ \ \ \ \ \ \ \ \ \ \ \ \ \ \ \ \ \ \ \ \ \ \ \ \ \ \ \ \ \ \ \ \ \ \ \ \ \ \ \ \ \ \ \ \ \ \ \ \ \ \ \ \ \ \ \ \ \ \ \ \ \ \ \ \ \ \ \ \ \ \ \ \ \ \ \ \ \ \ \ \ \ \ \ \ \ \ \ \ \ \ \ \ \ \ \ \ \ \ \ \ \ \ \ \ \ \ \ \ \ \ \ \ \ \ \ \ \ \ \ \ \ \ \ \ \ \ \ \ \ \ \ \ \ \ \ \ \ \ \ \ \ \ \ \ \ \ \ \ \ \ \ \ \ \ \ \ \ \ \ \ \ \ \ \ \ \ \ \ \ \ \ \ \ \ \ \ \ \ \ \ \ \ \ \ \ \ \ \ \ \ \ \ \ \ \ \ \ \ \ \ \ \ \ \ \ \ \ \ \ \ \ \ \ \ \ \ \ \ \ \ \ \ \ \ \ \ \ \ \ \ \ \ \ \ \ \ \ \ \ \ \ \ \ \ \ \ \ \ \ \ \ \ \ \ \ \ \ \ \ \ \ \ \ \ \ \ \ \ \ \ \ \ \ \ \ \ \ \ \ \ \ \ \ \ \ \ \ \ \ \ \ \ \ \ \ \ \ \ \ \ \ \ \ \ \ \ \ \ \ \ \ \ \ \ \ \ \ \ \ \ \ \ \ \ \ \ \ \ \ \ \ \ \ \ \ \ \ \ \ \ \ \ \ \ \ \ \ \ \ \ \ \ \ \ \ \ \ \ \ \ \ \ \ \ \ \ \ \ \ \ \ \ \ \ \ \ \ \ \ \ \ \ \ \ \ \ \ \ \ \ \ \ \ \ \ \ \ \ \ \ \ \ \ \ \ \ \ \ \ \ \ \ \ \ \ \ \ \ \ \ \ \ \ \ \ \ \ \ \ \ \ \ \ \ \ \ \ \ \ \ \ \ \ \ \ \ \ \ \ \ \ \ \ \ \ \ \ \ \ \ \ \ \ \ \ \ \ \ \ \ \ \ \ \ \ \ \ \ \ \ \ \ \ \ \ \ \ \ \ \ \ \ \ \ \ \ \ \ \ \ \ \ \ \ \ \ \ \ \ \ \ \ \ \ \ \ \ \ \ \ \ \ \ \ \ \ \ \ \ \ \ \ \ \ \ \ \ \ \ \ \ \ \ \ \ \ \ \ \ \ \ \ \ \ \ \ \ \ \ \ \ \ \ \ \ \ \ \ \ \ \ \ \ \ \ \ \ \ \ \ \ \ \ \ \ \ \ \ \ \ \ \ \ \ \ \ \ \ \ \ \ \ \ \ \ \ \ \ \ \ \ \ \ \ \ \ \ \ \ \ \ \ \ \ \ \ \ \ \ \ \ \ \ \ \ \ \ \ \ \ \ \ \ \ \ \ \ \ \ \ \ \ \ \ \ \ \ \ \ \ \ \ \ \ \ \ \ \ \ \ \ \ \ \ \ \ \ \ \ \ \ \ \ \ \ \ \ \ \ \ \ \ \
\begin{equation}
v_{1}=\frac{2}{\rho_{0}^{3}}(-\rho_{0}^{2}v_{0}\rho_{1}+2u_{1}) \label{v1x}%
\end{equation}

\end{lemma}

\begin{proof}
Starting from the expression (\ref{v1}), let us further reduce the following
quantity
\begin{align*}
R_{1}  &  =\rho_{0}^{3}\frac{v_{0}v_{1}}{2}=\rho_{0}^{3}\left[  \varphi
_{1}\varphi_{2}+\frac{1}{4}f_{0}\varphi_{1}\theta_{1}^{2}+g_{0}\theta
_{1}\theta_{2}\right] \\
&  =\varphi_{1}(\rho_{0}^{3}\varphi_{2})+\theta_{1}(g_{0}\rho_{0}^{3}%
\theta_{2})+\frac{1}{4}f_{0}\rho_{0}^{3}\varphi_{1}\theta_{1}^{2}%
\end{align*}
in terms of temporal invariants of lower order, as follows. In the last
expression for $R_{1}$, the second order invariants $\varphi_{2}$ and
$\theta_{2}$ can be reduced using the identities (\ref{E0}), and the
invariants $\varphi_{1}$ and $\theta_{1}$ can be "reduced" to $v_{0}$ by
(\ref{3.4}): \
\begin{align*}
R_{1}  &  =\varphi_{1}[-\rho_{0}^{2}\rho_{1}\varphi_{1}+\omega a_{0}\rho
_{0}\theta_{1}+\frac{1}{4}f_{0}\rho_{0}^{3}\theta_{1}^{2}+2\mu_{0}]\\
&  +\theta_{1}[-g_{0}\rho_{0}^{2}\rho_{1}\theta_{1}-\omega a_{0}\rho
_{0}\varphi_{1}-\frac{1}{2}f_{0}\rho_{0}^{3}\varphi_{1}\theta_{1}+2\eta
_{0}]+\frac{1}{4}f_{0}\rho_{0}^{3}\varphi_{1}\theta_{1}^{2}\\
&  =-\frac{1}{4}f_{0}\rho_{0}^{3}\varphi_{1}\theta_{1}^{2}-\rho_{0}^{2}%
\rho_{1}(\varphi_{1}^{2}+g_{0}\theta_{1}^{2})+2(\mu_{0}\varphi_{1}+\eta
_{0}\theta_{1})+\frac{1}{4}f_{0}\rho_{0}^{3}\varphi_{1}\theta_{1}^{2}\\
&  =-\rho_{0}^{2}\rho_{1}v_{0}^{2}+2v_{0}(\mu_{0}J_{\varphi}+\eta_{0}%
J_{\theta})=-\rho_{0}^{2}\rho_{1}v_{0}^{2}+2v_{0}u_{1},
\end{align*}
where at the end we have used the relation $u_{1}=\mu_{0}J_{\varphi}+\eta
_{0}J_{\theta}$. Recall that $u_{1}=\tau_{0}$ is the tangential derivative
$U_{\tau}^{\ast}$ of $U^{\ast}$ at $t=0$, which is an intrinsic invariant of
order 1. Comparison of the two expressions for $R_{1}$ yields the identity
(\ref{v1x}).
\end{proof}

Now, to calculate the reductions of $J_{\varphi}^{\prime},J_{\theta}^{\prime}$
from their expressions in (\ref{Jprime1}), we substitute the expression
(\ref{v1x}) for $v_{1}$ and the expressions (\ref{E0}) for $\varphi_{2}$ and
$\theta_{2}$, which yields%
\begin{align}
J_{\varphi}^{\prime}  &  =\frac{1}{v_{0}^{2}}(2\varphi_{2}-v_{0}%
v_{1}J_{\varphi})=\frac{1}{2}f_{0}J_{\theta}^{2}+\frac{2\omega\sin\varphi_{0}%
}{\rho_{0}^{2}v_{0}}J_{\theta}+\frac{4}{\rho_{0}^{3}v_{0}^{2}}(\mu_{0}%
-u_{1}J_{\varphi})\label{J3}\\
J_{\theta}^{\prime}  &  =\frac{1}{v_{0}^{2}}(2\theta_{2}-v_{1}J_{\theta
})=-\frac{f_{0}}{g_{0}}J_{\varphi}J_{\theta}-\frac{2\omega\sin\varphi_{0}%
}{g_{0}\rho_{0}^{2}v_{0}}J_{\varphi}+\frac{4}{\rho_{0}^{3}v_{0}^{2}}%
(\frac{\eta_{0}}{g_{0}}-u_{1}J_{\theta})\nonumber
\end{align}
This enables us to reduce the following intrinsic invariant
\ \ \ \ \ \ \ \ \ \ \
\begin{align*}
&  U_{\varphi}^{\ast}J_{\varphi}^{\prime}+U_{\theta}^{\ast}J_{\theta}^{\prime
}\\
&  =U_{\varphi}^{\ast}[\frac{1}{2}f_{0}J_{\theta}^{2}+\frac{2\omega\sin
\varphi_{0}}{\rho_{0}^{2}v_{0}}J_{\theta}+\frac{4}{\rho_{0}^{3}v_{0}^{2}}%
(\mu_{0}-u_{1}J_{\varphi})]\\
&  +U_{\theta}^{\ast}[-\frac{f_{0}}{g_{0}}J_{\varphi}J_{\theta}-\frac
{2\omega\sin\varphi_{0}}{g_{0}\rho_{0}^{2}v_{0}}J_{\varphi}+\frac{4}{\rho
_{0}^{3}v_{0}^{2}}(\frac{\eta_{0}}{g_{0}}-u_{1}J_{\theta})]\\
&  =f_{0}J_{\theta}(\frac{1}{2}U_{\varphi}^{\ast}J_{\theta}-\frac{U_{\theta
}^{\ast}}{g_{0}}J_{\varphi})+\frac{2\omega\sin\varphi_{0}}{\rho_{0}^{2}v_{0}%
}(U_{\varphi}^{\ast}J_{\theta}-\frac{U_{\theta}^{\ast}}{g_{0}}J_{\varphi})\\
&  +\frac{4}{\rho_{0}^{3}v_{0}^{2}}[(\mu_{0}^{2}+\frac{\eta_{0}^{2}}{g_{0}%
})-u_{1}(U_{\varphi}^{\ast}J_{\varphi}+U_{\theta}^{\ast}J_{\theta})]\\
&  =-\frac{1}{2}f_{0}U_{\varphi}^{\ast}J_{\theta}^{2}-\frac{f_{0}}{\sin
\varphi_{0}}w_{0}J_{\theta}-\frac{2\omega w_{0}}{\rho_{0}^{2}v_{0}}+\frac
{4}{\rho_{0}^{3}v_{0}^{2}}(\tau_{0}^{2}+w_{0}^{2}-\tau_{0}^{2})\\
&  =\frac{4w_{0}^{2}}{\rho_{0}^{3}v_{0}^{2}}-\frac{2\omega w_{0}}{\rho_{0}%
^{2}v_{0}}-f_{0}J_{\theta}(\frac{w_{0}}{\sin\varphi_{0}}+\frac{1}{2}%
U_{\varphi}^{\ast}J_{\theta})\\
&  =w_{0}K_{0}-f_{0}J_{\theta}(\frac{w_{0}}{\sin\varphi_{0}}+\frac{1}%
{2}U_{\varphi}^{\ast}J_{\theta})
\end{align*}
and hence the expression for $\tau_{1}$ in (\ref{first coeff}) reduces to%
\begin{equation}
\tau_{1}=w_{0}K_{0}+J \label{t1 redu}%
\end{equation}
where
\[
J=-f_{0}J_{\theta}(\frac{w_{0}}{\sin\varphi_{0}}+\frac{1}{2}U_{\varphi}^{\ast
}J_{\theta})+(U_{\varphi\varphi}^{\ast}J_{\varphi}^{2}+2U_{\varphi\theta
}^{\ast}J_{\varphi}J_{\theta}+U_{\theta\theta}^{\ast}J_{\theta}^{2})
\]
is a first order intrinsic invariant. Similarly, the expression for $w_{1}$ in
(\ref{first coeff}) reduces to
\begin{align}
w_{1}  &  =\frac{2u_{1}}{\rho_{0}^{2}v_{0}}(\omega-\frac{2w_{0}}{\rho_{0}%
v_{0}})+\eta_{0}\cos\varphi_{0}(J_{\varphi}^{2}+J_{\theta}^{2})+(\frac
{U_{\theta\theta}^{\ast}}{\sin\varphi_{0}}-U_{\varphi\theta}^{\ast}\sin
\varphi_{0})J_{\theta}^{2}\label{omega1x}\\
&  +(\mu_{0}\cos\varphi_{0}+\frac{U_{\theta\theta}^{\ast}}{\sin\varphi_{0}%
}-U_{\varphi\varphi}^{\ast}\sin\varphi_{0})J_{\varphi}J_{\theta}\nonumber
\end{align}

It would be rather laborious to calculate a complete reduction of $K_{1}$
starting from its order $3$ expression in (\ref{first coeff}). However, a
reduction is given implicitly by the second curvature equation $Eq2$
calculated below (see (\ref{basic system}), where the occurrence of $K_{1}$
lies in the coefficients $J_{2}$ and $J_{4}$. Solving the equation with
respect to $K_{1}$ yields the first of the following three expressions \
\begin{align}
K_{1}  &  =2K_{0}\frac{2K_{0}u_{1}\rho_{0}^{2}v_{0}-w_{1}\rho_{0}^{2}%
v_{0}-w_{0}\rho_{0}\rho_{1}+2\omega u_{1}}{\rho_{0}^{2}v_{0}(\omega\rho
_{0}v_{0}-2w_{0})}\label{K1x}\\
&  =-4\frac{2K_{0}u_{1}\rho_{0}^{2}v_{0}-w_{1}\rho_{0}^{2}v_{0}-w_{0}\rho
_{0}\rho_{1}+2\omega u_{1}}{\rho_{0}^{5}v_{0}^{3}}\nonumber\\
&  =\frac{16u_{1}}{\rho_{0}^{3}v_{0}^{2}}(\omega-\frac{2w_{0}}{\rho_{0}v_{0}%
})+\frac{4w_{1}}{\rho_{0}^{3}v_{0}^{2}}+\frac{4\ }{\rho_{0}^{5}v_{0}^{3}}%
(\rho_{0}\rho_{1}w_{0}-2\omega u_{1})\ \nonumber
\end{align}
and then elimination of $K_{0}$ twice using (\ref{K0}) yields the other two
expressions. Consequently, a complete reduction follows by substituting the
reduced expression (\ref{omega1x}) for $w_{1}$, but we shall not need this and
hence it is omitted.

\subsection{Derivation of the basic algebraic system of equations}

According to Proposition \ref{BasicProp}, the crucial problem indicated as
follows
\begin{equation}
\lbrack\text{intrinsic data}]_{p}\text{ }\rightarrow(\rho_{0},\rho_{1},v_{0}),
\label{2nd step}%
\end{equation}
is to calculate the basic temporal invariants $(\rho_{0},\rho_{1},v_{0})$ from
intrinsic data representing the local relative geometry of the pair
$(\gamma^{\ast},\nabla U^{\ast})$. The intrinsic data that we shall utilize
consist of the direction element $(J_{\varphi},J_{\theta})_{p}$ and the
following six \emph{basic intrinsic invariants}%

\begin{equation}
u_{0},u_{1},w_{0};w_{1},K_{0},K_{1}\text{ }(w_{0}\neq0,K_{0}\neq0),
\label{sixtuple}%
\end{equation}
also referred to as the \emph{basic 6-tuple} at $p$. As indicated, the
non-vanishing of $w_{0}$ and $K_{0}$ is always assumed in the sequel to avoid
cases of singular behavior. Note that $w_{0}\neq0$ means that $\gamma^{\ast}$
is transversal to the gradient flow of $U^{\ast}$ at $p$. We remark that by
(\ref{t1 redu}) we do not need the second order invariant $\tau_{1}$ in
(\ref{sixtuple}) since it can be derived from the others, cf. (\ref{t1 redu}).

The basic 6-tuple (\ref{sixtuple}) is the union of two triples of intrinsic
invariants, namely the \emph{basic} $U^{\ast}$\emph{-invariants} $(u_{0}%
,\tau_{0},w_{0})$ depending only on $U^{\ast}$ and $(J_{\varphi},J_{\theta
})_{p}$
\ \ \ \ \ \ \ \ \ \ \ \ \ \ \ \ \ \ \ \ \ \ \ \ \ \ \ \ \ \ \ \ \ \ \ \ \ \ \ \ \ \ \ \ \ \ \ \ \ \ \ \ \ \ \ \ \ \ \ \ \ \ \ \ \ \ \ \ \ \ \ \ \ \ \ \ \ \ \ \ \ \ \ \ \ \ \ \ \ \ \ \ \ \ \ \ \ \ \ \ \ \ \ \ \ \ \ \ \ \ \ \ \ \ \ \ \ \ \ \ \ \ \ \ \ \ \ \ \ \ \ \ \ \ \ \ \ \ \ \ \ \ \ \ \ \ \ \ \ \ \ \ \ \ \ \ \ \ \ \ \ \ \ \ \ \ \ \ \ \ \ \ \ \ \ \ \ \ \ \ \ \ \ \ \ \ \ \ \ \ \ \ \ \ \ \ \ \ \ \ \ \ \ \ \ \ \ \ \ \ \ \ \ \ \ \ \ \ \ \ \ \ \ \ \ \ \ \ \ \ \ \ \ \ \ \ \ \ \ \ \ \ \ \ \ \ \ \ \ \ \ \ \ \ \ \ \ \ \ \ \ \ \ \ \ \ \ \ \ \ \ \ \ \ \ \ \ \ \ \ \ \ \ \ \ \ \ \ \ \ \ \ \ \ \ \ \ \ \ \ \ \ \ \ \ \ \ \ \ \ \ \ \ \ \ \ \ \ \ \ \ \ \ \ \ \ \ \ \ \ \ \ \ \ \ \ \ \ \ \ \ \ \ \ \ \ \ \ \ \ \ \ \ \ \ \ \ \ \ \ \ \ \ \ \ \ \ \ \ \ \ \ \ \ \ \ \ \ \ \ \ \ \ \ \ \ \ \ \ \ \ \ \ \ \ \ \ \ \ \ \ \ \ \ \ \ \ \ \ \ \ \ \ \ \ \ \ \ \ \ \ \ \ \ \ \ \ \ \ \ \ \ \ \ \ \ \ \ \ \ \ \ \ \ \ \ \ \ \ \ \ \ \ \ \ \ \ \ \ \ \ \ \ \ \ \ \ \ \ \ \ \ \ \ \ \ \ \ \ \ \ \ \ \ \ \ \ \ \ \ \ \ \ \ \ \ \ \ \ \ \ \ \ \ \ \ \ \ \ \ \ \ \ \ \ \ \ \ \ \ \ \ \ \ \ \ \ \ \ \ \ \ \ \ \ \ \ \ \ \ \ \ \ \ \ \ \ \ \ \ \ \ \ \ \ \ \ \ \ \ \ \ \ \ \ \ \ \ \ \ \ \ \ \ \ \ \ \ \ \ \ \ \ \ \ \ \ \ \ \ \ \ \ \ \ \ \ \ \ \ \ \ \ \ \ \ \ \ \ \ \ \ \ \ \ \ \ \ \ \ \ \ \ \ \ \ \ \ \ \ \ \ \ \ \ \ \ \ \ \ \ \ \ \ \ \ \ \ \ \ \ \ \ \ \ \ \ \ \ \ \ \ \ \ \ \ \ \ \ \ \ \ \ \ \ \ \ \ \ \ \ \ \ \ \ \ \ \ \ \ \ \ \ \ \ \ \ \ \ \ \ \ \ \ \ \ \ \ \ \ \ \ \ \ \ \ \ \ \ \ \ \ \ \ \ \ \ \ \ \ \ \ \ \ \ \ \ \ \ \ \ \ \ \ \ \ \ \ \ \ \ \ \ \ \ \ \ \ \ \ \ \ \ \ \ \ \ \ \ \ \ \ \ \ \ \ \ \ \ \ \ \ \ \ \ \ \ \ \ \ \ \ \ \ \ \
\begin{equation}
u_{0}=U^{\ast}(p),\tau_{0}=u_{1}=U_{\mathbf{\tau}}^{\ast}(p),w_{0}%
=U_{\mathbf{\nu}}^{\ast}(p)\text{,} \label{U*-invariants}%
\end{equation}
and the \emph{basic curvature invariants} $(w_{1},K_{0},K_{1})$ reflecting
some of the local geometry of $(\gamma^{\ast},\nabla U^{\ast})$ at $p$.

Now, we shall set up a system of three algebraic equations used to calculate
$(\rho_{0},\rho_{1},v_{0})$. The coefficients of the system are rational
functions of the invariants in (\ref{sixtuple}), together with the parameters
$(\omega,h)$. For notational convenience, let us introduce the following four
intrinsic invariants
\begin{align}
J_{1}  &  =2u_{0}-\mathfrak{S}_{0}\text{, }J_{2}=\frac{2u_{1}-\mathfrak{S}%
_{1}}{\mathfrak{S}_{0}}\label{J-symbols1}\\
\text{ }J_{3}  &  =\frac{2u_{1}}{w_{0}}\text{, \ }J_{4}=\frac{-K_{1}}%
{2K_{0}w_{0}}\nonumber
\end{align}
Our first equation is
\begin{equation}
Eq1:\rho_{0}^{3}v_{0}^{2}+2\omega\frac{v_{0}\rho_{0}}{K_{0}}=4\mathfrak{S}%
_{0}, \label{Eq1}%
\end{equation}
which is just the leftmost identity in (\ref{K0}). Another equation is derived
from the energy integral (\ref{h-int}) evaluated at time $t=0$, which combined
with $Eq1$ yields
\begin{equation}
Eq3:\rho_{0}(\rho_{1}^{2}-2h)+\frac{\omega^{2}}{\rho_{0}}-\frac{\omega}%
{2K_{0}}\rho_{0}v_{0}=2u_{0}-\mathfrak{S}_{0}\ \label{Eq3}%
\end{equation}
Finally, to obtain last equation $Eq2$, let us differentiate the curvature
equation (\ref{K*}) and evaluate at $t=0$, which yields the identity%
\begin{equation}
2\rho_{0}^{3}v_{0}v_{1}+3\rho_{0}^{2}v_{0}^{2}\rho_{1}-4v_{0}\mathfrak{S}%
_{1}+\frac{2\omega}{K_{0}}[v_{0}\rho_{1}+\rho_{0}v_{1}-\frac{K_{1}}{K_{0}}%
\rho_{0}v_{0}^{2}]=0 \label{Eq2'}%
\end{equation}
Here we replace the temporal invariant $v_{1}$ by its reduced expression
(\ref{v1x}) and make use of (\ref{Eq1}) to calculate
\ \ \ \ \ \ \ \ \ \ \ \ \ \ \ \ \ \ \ \ \ \ \ \ \ \ \ \ \ \ \ \ \ \ \ \ \ \ \
\begin{align*}
0  &  =[4v_{0}+\frac{4\omega}{K_{0}}\frac{1}{\rho_{0}^{2}}](\frac{1}{2}%
\rho_{0}^{3}v_{1})+3\rho_{0}^{2}v_{0}^{2}\rho_{1}-4v_{0}\mathfrak{S}_{1}\\
&  +\frac{2\omega}{K_{0}}v_{0}\rho_{1}-\frac{2\omega K_{1}}{K_{0}^{2}}\rho
_{0}v_{0}^{2}\\
&  =(4v_{0}+\frac{4\omega}{K_{0}}\frac{1}{\rho_{0}^{2}})(-\rho_{0}^{2}%
v_{0}\rho_{1}+2u_{1})+3\rho_{0}^{2}v_{0}^{2}\rho_{1}-4v_{0}\mathfrak{S}_{1}\\
&  +\frac{2\omega}{K_{0}}v_{0}\rho_{1}-\frac{2\omega K_{1}}{K_{0}^{2}}\rho
_{0}v_{0}^{2}\\
&  =-\rho_{0}^{2}v_{0}^{2}\rho_{1}+4(2u_{1}-\mathfrak{S}_{1})v_{0}%
+\frac{2\omega}{K_{0}}(4u_{1}\frac{1}{\rho_{0}^{2}}-\frac{K_{1}}{K_{0}}%
\rho_{0}v_{0}^{2})\\
&  =-\frac{\rho_{1}}{\rho_{0}}(\rho_{0}^{3}v_{0}^{2}+\frac{2\omega}{K_{0}}%
\rho_{0}v_{0})+4(2u_{1}-\mathfrak{S}_{1})v_{0}+\frac{2\omega}{K_{0}}%
(4u_{1}\frac{1}{\rho_{0}^{2}}-\frac{K_{1}}{K_{0}}\rho_{0}v_{0}^{2})
\end{align*}
Using $Eq1$ to substitute the first term in the last line, we further simplify
the last identity as follows:
\begin{align}
\mathfrak{S}_{0}\frac{\rho_{1}}{\rho_{0}}  &  =(2u_{1}-\mathfrak{S}_{1}%
)v_{0}+\frac{\omega}{K_{0}}(2u_{1}\frac{1}{\rho_{0}^{2}}-\frac{K_{1}}{2K_{0}%
}\rho_{0}v_{0}^{2})\nonumber\\
\frac{\rho_{1}}{\rho_{0}v_{0}}  &  =\frac{2u_{1}-\mathfrak{S}_{1}%
}{\mathfrak{S}_{0}}+\omega\lbrack\frac{2u_{1}}{K_{0}\mathfrak{S}_{0}}\frac
{1}{\rho_{0}^{2}v_{0}}-\frac{K_{1}}{2K_{0}^{2}\mathfrak{S}_{0}}\rho_{0}%
v_{0}]\nonumber\\
Eq2  &  :\frac{\rho_{1}}{\rho_{0}v_{0}}=J_{2}+\omega\lbrack J_{3}\frac{1}%
{\rho_{0}^{2}v_{0}}+J_{4}\rho_{0}v_{0}] \label{eq3}%
\end{align}
\ \ \ \ \ \ \ \ 

\begin{definition}
\label{BasicSystem}For a given energy-momentum pair $(h,\omega)$, the
\underline{basic} \underline{algebraic} \underline{system,} affiliated with
the direction element $(J_{\varphi},J_{\theta})_{p}$ and the basic curvature
invariants $(w_{1},K_{0},K_{1})$ at $p$, consists of the following three
algebraic equations with the basic temporal invariants $(\rho_{0},\rho
_{1},v_{0})$ as the variables:
\begin{align}
Eq1  &  :\rho_{0}^{3}v_{0}^{2}+2\omega\frac{\rho_{0}v_{0}}{K_{0}%
}=4\mathfrak{S}_{0}\nonumber\\
Eq2  &  :\frac{\rho_{1}}{\rho_{0}v_{0}}-\omega\lbrack J_{3}\frac{1}{\rho
_{0}^{2}v_{0}}+J_{4}\rho_{0}v_{0}]\ =J_{2}\ \label{basic system}\\
Eq3  &  :\rho_{0}(\rho_{1}^{2}-2h)+\frac{\omega^{2}}{\rho_{0}}-\frac{\omega
}{2K_{0}}\rho_{0}v_{0}=J_{1}\nonumber
\end{align}

\end{definition}

\begin{remark}
The angular momentum and energy constants are, in fact, expressible at the
moduli space level, namely in terms of intrinsic invariants and the basic
temporal invariants, as follows:%
\begin{align}
\omega &  =\frac{2w_{0}}{\rho_{0}v_{0}}-\frac{1}{2}K_{0}\rho_{0}^{2}%
v_{0}\label{omega2}\\
h  &  =\frac{1}{2}\rho_{1}^{2}+\frac{1}{8}\rho_{0}^{2}v_{0}^{2}+2\frac
{w_{0}^{2}}{\rho_{0}^{4}v_{0}^{2}}-\frac{w_{0}K_{0}+u_{0}}{\rho_{0}}+\frac
{1}{8}K_{0}^{2}\rho_{0}^{2}v_{0}^{2} \label{h2}%
\end{align}
The expression for $\omega$ is simply a restatement of $Eq1$ in
(\ref{basic system}), whereas the expression for $h$ follows from $Eq3$ or
more simply from (\ref{S7a}), by inserting the expression for $\omega$.
\end{remark}

\subsection{Solution of the basic algebraic system of equations
\ \ \ \ \ \ \ \ \ \ \ \ \ \ \ \ \ \ \ \ \ \ \ \ \ \ \ \ \ \ \ \ \ \ \ \ \ \ \ \ \ \ \ \ \ \ \ \ \ \ \ \ \ \ \ \ \ \ \ \ \ \ \ \ \ \ \ \ \ \ \ \ \ \ \ \ \ \ \ \ \ \ \ \ \ \ \ \ \ \ \ \ \ \ \ \ \ \ \ \ \ \ \ \ \ \ \ \ \ \ \ \ \ \ \ \ \ \ \ \ \ \ \ \ \ \ \ \ \ \ \ \ \ \ \ \ \ \ \ \ \ \ \ \ \ \ \ \ \ \ \ \ \ \ \ \ \ \ \ \ \ \ \ \ \ \ \ \ \ \ \ \ \ \ \ \ \ \ \ \ \ \ \ \ \ \ \ \ \ \ \ \ \ \ \ \ \ \ \ \ \ \ \ \ \ \ \ \ \ \ \ \ \ \ \ \ \ \ \ \ \ \ \ \ \ \ \ \ \ \ \ \ \ \ \ \ \ \ \ \ \ \ \ \ \ \ \ \ \ \ \ \ \ \ \ \ \ \ \ \ \ \ \ \ \ \ \ \ \ \ \ \ \ \ \ \ \ \ \ \ \ \ \ \ \ \ \ \ \ \ \ \ \ \ \ \ \ \ \ \ \ \ \ \ \ \ \ \ \ \ \ \ \ \ \ \ \ \ \ \ \ \ \ \ \ \ \ \ \ \ \ \ \ \ \ \ \ \ \ \ \ \ \ \ \ \ \ \ \ \ \ \ \ \ \ \ \ \ \ \ \ \ \ \ \ \ \ \ \ \ \ \ \ \ \ \ \ \ \ \ \ \ \ \ \ \ \ \ \ \ \ \ \ \ \ \ \ \ \ \ \ \ \ \ \ \ \ \ \ \ \ \ \ \ \ \ \ \ \ \ \ \ \ \ \ \ \ \ \ \ \ \ \ \ \ \ \ \ \ \ \ \ \ \ \ \ \ \ \ \ \ \ \ \ \ \ \ \ \ \ \ \ \ \ \ \ \ \ \ \ \ \ \ \ \ \ \ \ \ \ \ \ \ \ \ \ \ \ \ \ \ \ \ \ \ \ \ \ \ \ \ \ \ \ \ \ \ \ \ \ \ \ \ \ \ \ \ \ \ \ \ \ \ \ \ \ \ \ \ \ \ \ \ \ \ \ \ \ \ \ \ \ \ \ \ \ \ \ \ \ \ \ \ \ \ \ \ \ \ \ \ \ \ \ \ \ \ \ \ \ \ \ \ \ \ \ \ \ \ \ \ \ \ \ \ \ \ \ \ \ \ \ \ \ \ \ \ \ \ \ \ \ \ \ \ \ \ \ \ \ \ \ \ \ \ \ \ \ \ \ \ \ \ \ \ \ \ \ \ \ \ \ \ \ \ \ \ \ \ \ \ \ \ \ \ \ \ \ \ \ \ \ \ \ \ \ \ \ \ }%

We may as well consider the system of rational equations (\ref{basic system})
from a purely algebraic point of view, with the three real variables
$(\rho_{0},\rho_{1},v_{0})$.The coefficients of the equations are themselves
rational expressions involving eight real constants (parameters), namely a
pair $(h,\omega)$ and a 6-tuple (\ref{sixtuple}). In Definition
\ref{BasicSystem} we also say that the system (\ref{basic system}) is
affiliated with a given 6-tuple (\ref{sixtuple}), irrespective of the
interpretation of the latter. Anyhow, we shall be interested only in
\underline{admissible} solutions $(\rho_{0},\rho_{1},v_{0})$, meaning that
$\rho_{0}>0$ and $v_{0}>0$.

\subsubsection{The special case of vanishing angular momentum}

For comparison reasons, let us recall the special case of \ three-body motions
with $\omega=0$ discussed in \cite{HS-2}, in which case the equations
(\ref{basic system}) simplify to
\begin{equation}
\rho_{0}^{3}v_{0}^{2}=4\mathfrak{S}_{0},\text{ }\frac{\rho_{1}}{\rho_{0}v_{0}%
}=J_{2},\text{\ }\rho_{0}(\rho_{1}^{2}-2h)=J_{1}\text{ }
\label{basic system 0}%
\end{equation}
Then there can be at most one admissible solution, uniquely given by
\begin{equation}
v_{0}=2\sqrt{\frac{\mathfrak{S}_{0}}{\rho_{0}^{3}}}\text{, \ }\rho_{1}%
=2J_{2}\text{\ }\sqrt{\frac{\mathfrak{S}_{0}}{\rho_{0}}}\text{, \ }\rho
_{0}=\frac{1}{2h}(4J_{2}^{2}\mathfrak{S}_{0}-J_{1})\text{ when }h\neq0.
\label{solution-0}%
\end{equation}

For completeness, let us also recall the special case of $(h,\omega)=(0,0)$.
Then the scaling symmetry of 3-body motions, which scales $\rho$ and the time
$t$, keeps $(h,\omega)=(0,0)$ invariant and does not affect the geometric
shape curve. Thus we can choose $\rho_{0}$ freely, and the above expressions
for $v_{0}$ and $\rho_{1}$ still hold. In any case, it follows from the
formulae (\ref{solution-0}) that positivity of $\rho_{0}$ and $v_{0}$ poses
some obvious conditions on the basic intrinsic invariants (\ref{sixtuple}).

\ \ \ 

\subsubsection{The general case $\omega\neq0$}

To study the basic algebraic system (\ref{basic system}) in general, let us
for convenience introduce the variables%

\[
x=\rho_{0},\text{ \ }y=\rho_{0}v_{0},\text{ \ }z=\rho_{1}%
\]
Then the system (\ref{basic system}) becomes%
\begin{align}
&  (i)\text{ \ }xy^{2}+\frac{2\omega}{K_{0}}y=4\mathfrak{S}_{0}\nonumber\\
&  (ii)\text{ \ }\frac{z}{y}-\omega\lbrack J_{3}\frac{1}{xy}+J_{4}%
y\ ]=J_{2}\label{xyz}\\
&  (iii)\text{ \ }x(z^{2}-2h)+\frac{\omega^{2}}{x}-\frac{\omega}{2K_{0}%
}y=J_{1}\nonumber
\end{align}

As an approach to solving the system, we propose to eliminate $x$ and $z$ and
seek an equation solely involving $y$, as follows. From equation (i)%
\begin{equation}
x=\frac{4\mathfrak{S}_{0}}{y^{2}}-\frac{2\omega}{K_{0}}\frac{1}{y}, \label{x}%
\end{equation}
which by substitution into equation (ii) yields%

\begin{align}
z  &  =J_{2}y+\omega y^{2}\left\{  \frac{J_{3}}{4\mathfrak{S}_{0}%
-\frac{2\omega}{K_{0}}y}+J_{4}\right\} \nonumber\\
&  =J_{2}y+\omega y^{2}\left\{  \frac{J_{5}+\omega J_{6}y}{4\mathfrak{S}%
_{0}-\frac{2\omega}{K_{0}}y}\right\} \nonumber\\
&  =\frac{J_{2}y(4\mathfrak{S}_{0}-\frac{2\omega}{K_{0}}y)+\omega y^{2}%
(J_{5}+\omega J_{6}y)}{4\mathfrak{S}_{0}-\frac{2\omega}{K_{0}}y}\nonumber\\
&  =\frac{y[4\mathfrak{S}_{0}J_{2}+\omega(J_{5}-\frac{2J_{2}}{K_{0}}%
)y+\omega^{2}J_{6}y^{2}]}{4\mathfrak{S}_{0}-\frac{2\omega}{K_{0}}y}.\nonumber
\end{align}
Thus, we arrive at the expression%
\begin{equation}
z=\frac{y(a+b\omega y+c\omega^{2}y^{2})}{d+e\omega y}, \label{z}%
\end{equation}
where we have introduced for notational convenience the intrinsic invariants
\begin{align}
J_{5}  &  =J_{3}+4\mathfrak{S}_{0}J_{4}=\frac{2u_{1}}{w_{0}}-\frac{2K_{1}%
}{K_{0}^{2}}\text{, \ }J_{6}=-\frac{2J_{4}}{K_{0}}=\frac{K_{1}}{K_{0}^{2}%
w_{0}}\label{J5-6}\\
a  &  =4\mathfrak{S}_{0}J_{2}\ ,\text{ }\ \text{ }b=J_{5}-\frac{2J_{2}}{K_{0}%
},\text{ }c=J_{6}\text{, }d=4\mathfrak{S}_{0}\text{, }e=-\frac{2}{K_{0}}
\label{abc}%
\end{align}
$\ $

Substitution of the expressions (\ref{x}) and (\ref{z}) for$\ x$ and $z$ into
equation (iii) yields%
\[
\frac{(a+\omega by+\omega^{2}cy^{2})^{2}}{d+\omega ey}-2h\frac{(d+\omega
ey)}{y^{2}}+\frac{\omega^{2}y^{2}}{(d+\omega ey)}+\frac{\omega e}{4}y=J_{1},
\]
which we state as the following polynomial equation of degree $\leq6$
\begin{align}
\Pi(y)=  &  y^{2}[\beta_{0}+\beta_{1}\omega y+\ \beta_{2}\omega^{2}%
y^{2}+\ \beta_{3}\omega^{3}y^{3}+\ \beta_{4}\omega^{4}y^{4}]\label{poly1}\\
&  +h[\alpha_{0}+\alpha_{1}\omega y+\alpha_{2}\omega^{2}y^{2}]=0\nonumber
\end{align}
where\textbf{ }the role of $(h,\omega)$ is made more transparent, and
$\alpha_{i},\beta_{i}$ are the following intrinsic invariants
\begin{align}
\alpha_{0}  &  =-2d^{2},\beta_{0}=a^{2}-J_{1}d,\alpha_{1}=-4de,\beta_{1}%
=\frac{de}{4}+2ab-J_{1}e,\label{coeff1}\\
\alpha_{2}  &  =-2e^{2},\beta_{2}=1+2ac+b^{2}+\frac{e^{2}}{4},\beta
_{3}=2bc,\beta_{4}=c^{2}\nonumber
\end{align}

We may also prefer to study the equation (\ref{poly1}) as a polynomial
equation in $Y=\omega y:$%
\begin{align}
P(Y)=  &  Y^{2}[\beta_{0}+\beta_{1}Y+\ \beta_{2}Y^{2}+\ \beta_{3}Y^{3}%
+\ \beta_{4}Y^{4}]\label{poly2}\\
&  +H[\alpha_{0}+\alpha_{1}Y+\alpha_{2}Y^{2}]=0\nonumber
\end{align}
with the explicit dependence on the parameters $(h,\omega)$ concentrated in
the single parameter
\[
H=h\omega^{2}\text{,}%
\]
as expected in view of Remark \ref{scaling}. For a given value of $H$, a root
$Y$ of the polynomial $P(Y)$ is said to be \underline{admissible} if the
following inequality holds%
\begin{equation}
sign(K_{0})Y<sign(K_{0})2w_{0} \label{admiss}%
\end{equation}

Now, let us fix the pair $(h,\omega)$ and hence also $H$. Since the polynomial
in (\ref{poly2}) has at most six roots, it is clear that the algebraic system
(\ref{xyz}) has at most six solutions $(x,y,z)$, obtained by expressing $x$
and $z$ in terms of $y=\omega^{-1}Y$ as explained above.

Moreover, a solution is real, namely $x,y,z$ are real, if and only if $Y$ \ is
a real root. Finally, a solution $(x,y,z)=(\rho_{0},\rho_{0}v_{0},\rho_{1})$
is admissible if and only if \ $x>0$ and $y>0$. In fact, by formula (\ref{x})
positivity of $x$ is simply the condition (\ref{admiss}) that $Y$ is
admissible, whereas positivity of $y$ demands $Y$ \ to have the same sign as
$\omega$.

For the applications involving shape curves, we prefer to regard the given
6-tuple $(u_{0},u_{1},w_{0};w_{1},K_{0},K_{1})$ as the basic 6-tuple at a
regular point $p$ of an oriented curve $\gamma^{\ast}$ on the shape sphere,
and hence the 6-tuple is the union of the $U^{\ast}$-triple $(u_{0}%
,u_{1},w_{0})$ and the basic curvature invariants $(w_{1},K_{0},K_{1})$ of
$\gamma^{\ast}$ at $p$. In fact, the direction element $(J_{\varphi}%
,J_{\theta})_{p}$ and the $U^{\ast}$-triple determine each other.

As an immediate consequence of the above observations we state the following
results, assuming $w_{0}\neq0,K_{0}\neq0$ as before.

\begin{proposition}
\label{A}For a given pair $(h,\omega)$, consider the basic algebraic system
(\ref{xyz}) affiliated with a direction element $(J_{\varphi},J_{\theta})_{p}$
and triple $(w_{1},K_{0},K_{1})$ of real numbers. Then there is a one-to-one
correspondence between the admissible solutions $(x,y,z)$ of the system and
solutions $\bar{\gamma}(t)$ of the ODE (\ref{NewtonRedu}) at the
energy-momentum level $(h,\omega)$, whose shape curves $\gamma^{\ast}$ passing
throught $p$ have i) the given direction $(J_{\varphi},J_{\theta})_{p}$, and
ii) the triple $(w_{1},K_{0},K_{1})$ as the basic curvature invariants at $p$.
\end{proposition}

\begin{proposition}
For a fixed value of $H=h\omega^{2}$, the negative admissible roots $Y$ of the
polynomial equation (\ref{poly2}) are in one-to one correspondence with the
admissible solutions $(x,y,z)$ in Proposition \ref{A}, for any fixed choice
$(h,\omega)$ with $\omega<0$. Similarly, the positive admissible roots $Y$
correspond to the admissible solutions $(x,y,z)$ when $\omega>0$.
\end{proposition}

Thus we end up with the following two cases:

\begin{itemize}
\item The basic algebraic system has no admissible solution $(\rho_{0}%
,\rho_{1},v_{0}).$ Then there is no 3-body motion $\gamma(t)$ whose shape
curve passes through $p$ with the given direction $(J_{\varphi},J_{\theta
})_{p}$ and with $(w_{1},K_{0},K_{1})$ as the basic curvature invariants.

\item The basic algebraic system at $p$ has an admissible solution $(\rho
_{0},\rho_{1},v_{0})$. As shown in Section 3.1, (\ref{1st reduction}), an
admissible solution together with the direction element $(J_{\varphi
},J_{\theta})_{p}$ yields the initial data (\ref{initial}) and hence a
solution $\bar{\zeta}(t)$ of the ODE (\ref{NewtonRedu}). As a consequence,
$\bar{\zeta}(t)$ is the moduli curve of a 3-body motion, whose shape curve
$\zeta^{\ast}$ has the direction $(J_{\varphi},J_{\theta})_{p}$ and basic
curvature invariants $(w_{1},K_{0},K_{1})$ at $p$. We remark that the number
of admissible solutions cannot exceed 6; in the examples we have studied so
far the number is at most 2.
\end{itemize}

\subsection{ Case studies with examples}

Without specifying the direction element $(J_{\varphi},J_{\theta})_{p}$ and
the (mass dependent) potential function $U^{\ast}$ on the shape sphere, we
shall start from a list of six numbers viewed as the basic 6-tuple
(\ref{sixtuple}) at $p$ of the shape curve $\gamma^{\ast}$ of a ficticious
three-body motion, for a given energy-momentum pair $(h,\omega)$. Then we
inquire whether the affiliated basic algebraic system (\ref{basic system}) has
admissible solutions $(\rho_{0},\rho_{1},v_{0})$. For each such solution we
know there is a solution $\overline{\gamma}(t)$ of the ODE (\ref{NewtonRedu})
in the moduli space $\bar{M}$ realizing the same basic 6-tuple (\ref{sixtuple}).\ 

Note that once an admissible root $Y$ of the polynomial $P(Y)$ is calculated,
one obtains the solution $(x,y,z)$ of the system (\ref{xyz}) by using the
formulas in (\ref{J-symbols1}), (\ref{x}) -- (\ref{abc}). We perform the
calculations using Maple, and high precision is needed.

\subsubsection{Example}

Let us choose (randomly) the following six numbers
\[
u_{0}=0.3,u_{1}=0.4,w_{0}=1.4,w_{1}=-0.3,K_{0}=0.4,K_{1}=0.2;\text{ }%
\]
and consider the basic algebraic system (\ref{xyz}) affiliated with this
6-tuple, and with $(h,\omega)$ unspecified. As before, we put $H=$
$h\omega^{2}$ and turn to the associated polynomial equation (\ref{poly2}).

Let us first choose $H=1$. The polynomial has four complex and two real roots%
\[
Y_{1}=-1.165697412\text{, }Y_{2}=1.521207930
\]
Since both $Y_{i}$ satisfy the inequality $Y<2w_{0}=2.8$, they are admissible
by (\ref{admiss}). Moreover, $\omega$ must have the same sign as $Y_{i}$, so
this yields the following triples of basic temporal invariants%
\begin{align*}
\text{Case 1}  &  \text{: }\omega<0;\text{ }\\
(\rho_{0},\rho_{1},v_{0})  &  =(14.59210391\cdot\omega^{2},-1.302577877\cdot
\omega^{-1},-0.07988549281\cdot\omega^{-3})\\
\text{Case 2}  &  \text{: }\omega>0;\\
\text{ }(\rho_{0},\rho_{1},v_{0})  &  =(2.763075661\cdot\omega^{2}%
,1.227863217\cdot\omega^{-1},0.5505487785\cdot\omega^{-3})
\end{align*}
In each case we are free to choose the energy-momentum pair $(h,\omega)$
subject to the constraint $h\omega^{2}=1$, and for each choice there is a
unique solution $\bar{\gamma}(t)$ of the ODE (\ref{NewtonRedu}) whose shape
curve $\gamma^{\ast}$ has basic 6-tuple $(u_{0},...,K_{1})$ as given above, at
any given direction element $(J_{\varphi},J_{\theta})_{p}$ compatible with the
$U^{\ast}$-invariants $(u_{0},u_{1},w_{0})$.

As long as $H>0$ the polynomial $P(Y)$ will have two admissible roots, and of
different sign, so the case $H=1$ is typical for energy $h>0$. But for
$H\leq0$ all roots are complex, hence there is no admissible solution in the
case of $h\leq0$.

\subsubsection{Example}

The following 6-tuple and associated polynomial (\ref{poly2}) are calculated:
\begin{align*}
\ u_{0}  &  =1,u_{1}=-0.01,w_{0}=0.2,w_{1}=0,K_{0}=0.1,K_{1}=0.0567\\
P(Y)  &  =0.00199-0.0241\cdot Y+0.0864\cdot Y^{2}-0.128\cdot Y^{3}+0.0804\cdot
Y^{4}=0
\end{align*}
Let us proceed as in the previous example. The polynomial has two admissible
roots $Y_{i}$, hence for fixed $\omega>0\ $\ we have an example with two
admissible solutions (basic temporal invariants)
\begin{align*}
Y_{1}  &  =0.137;\text{ \ }(\rho_{0},\rho_{1},v_{0})=(279.8\cdot\omega
^{2},0.049\cdot\omega^{-1},0.00049\cdot\omega^{-3})\\
\text{ \ }Y_{2}  &  =0.390;\text{ \ }(\rho_{0},\rho_{1},v_{0})=(0.510\cdot
\omega^{2},-0.198\cdot\omega^{-1},0.776\cdot\omega^{-3})
\end{align*}

\subsubsection{Example}

Consider the special case $h=0$, so the polynomial (\ref{poly1}) has the
factor $y^{2}$, which after cancelling yields an equation of degree 4
\begin{align}
P(y)  &  =\beta_{0}+\beta_{1}\omega y+\ \beta_{2}\omega^{2}y^{2}+\ \beta
_{3}\omega^{3}y^{3}+\ \beta_{4}\omega^{4}y^{4}=0,\label{deg4}\\
\beta_{0}  &  =a^{2}-J_{1}d=4\mathfrak{S}_{0}(4\mathfrak{S}_{0}J_{2}^{2}%
-J_{1}),\text{ \ }\nonumber
\end{align}
and equation (\ref{poly2}) with $H=0$ reads
\begin{equation}
P(Y)=\beta_{0}+\beta_{1}Y+\beta_{2}Y^{2}+\beta_{3}Y^{3}+\beta_{4}Y^{4}=0.
\label{F}%
\end{equation}
Moreover, $P(Y)$ becomes a quadratic polynomial if and only if $K_{1}=0$, by
(\ref{J5-6}), (\ref{abc}), (\ref{coeff1}).

Now, choose the following 6-tuple (\ref{sixtuple})
\[
\ u_{0}=1,u_{1}=-0.1,w_{0}=0.2,w_{1}=0,K_{0}=0.1,K_{1}=0,
\]
which yields a qadratic polynomial with two admissible roots
\begin{align*}
P(Y)  &  =6.4\cdot10^{-3}-0.416\cdot Y+1.02\cdot Y^{2}\\
Y_{1}  &  \approx0.016\text{, \ }Y_{2}\approx0.392
\end{align*}
Hence, for $\omega>0$ the algebraic system (\ref{basic system}) has the
following two admissible solutions $(\rho_{0},\rho_{1},v_{0}):$
\begin{align*}
\text{Case 1}  &  \text{: }\rho_{0}\approx295\cdot10^{4}\cdot\omega^{2}\text{,
}\rho_{1}\approx-0.0016\cdot\omega^{-1}\text{, }v_{0}\approx5.35\cdot
10^{-7}\omega^{-3}\text{ }\\
\text{Case 2}  &  \text{: }\rho_{0}\approx1.06\cdot\omega^{2}\text{, }\rho
_{1}\approx-0.979\cdot\omega^{-1}\text{, }v_{0}\approx0.368\cdot\omega^{-3}%
\end{align*}
As a control, we verify that $h=0$ by inserting these values into the energy
integral (using high precision arithmetic)
\[
h=\frac{1}{2}\rho_{1}^{2}+\frac{\rho_{0}^{2}}{8}v_{0}^{2}+\frac{\omega^{2}%
}{2\rho_{0}^{2}}-\frac{u_{0}}{\rho_{0}}%
\]

\subsubsection{Example (linear motions)}

Observe that in order to obtain an algebraic system with three independent
equations such as (\ref{xyz}), the nonvanishing of some of the $U^{\ast}%
$-invariants (\ref{U*-invariants}) is crucial. But for comparison reason,
consider the extreme case $U^{\ast}=0$ where this fails, namely linear 3-body
motions (in absence of forces). \ Now, $Eq1$ of the system (\ref{basic system}%
) reads%
\[
K^{\ast}=-\frac{2\omega}{\rho^{2}v},
\]
and this holds everywhere along the motion. Therefore, in the case $\omega=0$
the shape curve $\gamma^{\ast}$ has vanishing curvature and is therefore
confined to a geodesic circle on the 2-sphere. In the case $\omega\neq0$,
$K^{\ast}$ cannot vanish, but $Eq2$ on the form (\ref{Eq2'}) together with
(\ref{v1x}) amount to $K_{1}=0$, namely $K^{\ast}$ is a constant.
Consequently, $\gamma^{\ast}$ is confined to a circle which is not a geodesic
circle. Reconstruction of the moduli curve $\bar{\gamma}(t)$ from the shape
curve alone is clearly not possible; there is little information "stored" in
the constant function $K^{\ast}.$

On the other hand, the simple example with $U^{\ast}$ a nonzero constant can
be solved explicitly. Now equation (i) of the ODE system (\ref{basic system})
decouples from the other two and can be solved explicity, as a 1-dimensional
Kepler problem with initial values $\rho_{0},\rho_{1}$. Then substitution into
the energy integral, for given values of $(h,\omega)$, yields the value of
$v_{0}^{2}$, and consequently one obtains the triple $(\rho_{0},\rho_{1}%
,v_{0})$. As before, together with a given point and direction element on the
shape sphere this yields the data to solve the initial value problem in
$\bar{M}$.

\subsubsection{Example (Henon 2)}

We consider a concrete example with three bodies of unit mass in the plane at
initial positions ($x_{i},0)$ on the x-axis and initial velocities $(0,y_{i})$
in the y-axis direction, namely%

\begin{align}
x_{1}  &  =-1.0207041786,\text{ }x_{2}=2.0532718983,\text{ }x_{3}%
=-1.0325677197\label{data}\\
y_{1}  &  =9.1265693140,\text{ }y_{2}=0.0660238922,\text{ }y_{3}%
=-9.1925932061\nonumber\\
h  &  =-1.040039,\text{ }\omega=0.312013\nonumber
\end{align}

We infer from these data that the shape curve $\gamma^{\ast}$ starts out at a
point $p=(\varphi_{0},\theta_{0})$ on the equator circle $\varphi=\pi/2$ and
with initial velocity perpendicular to this circle. Moreover, we infer
\[
\rho_{1}=\theta_{1}=u_{1}=w_{1}=\beta_{3}=\beta_{4}=0
\]
Calculation of the basic 6-tuple of $\gamma^{\ast}$ at $p$ :
\[
(u_{0},u_{1,}w_{0},w_{1},K_{0},K_{1})=(213.6058,0,-31771.876,0,-75.30872,0)
\]

Now, we can determine the associated polynomial $P(Y)$, which has degree 4 and
has two negative and two positive roots $Y_{i}$. The positive roots are
admissible and there are two admissible solutions $(\rho_{0},\rho_{1},v_{0})$
of the basic algebraic system (\ref{basic system}), as follows:%
\begin{align*}
Y_{1}  &  =89.01744668991:\rho_{0}=0.02076,\rho_{1}=0,\text{ }v_{0}%
=13741.798\\
Y_{2}  &  =8.083161335258:\rho_{0}=2.51475,\rho_{1}=0,\text{ }v_{0}=10.30184
\end{align*}

\begin{remark}
In the above example, we observe that the original 3-body motion is one of the
two 3-body motions generated by the roots $Y_{i}$, namely by the root $Y_{2.}$
Calculation of the corresponding initial data set (\ref{initial}) for the ODE
(\ref{NewtonRedu}) can be checked to be the same as the initial data set
deduced from the information (\ref{data}). On the other hand, the root $Y_{1}$
yields another 3-body motion, but we claim its shape curve is different from
the shape curve of the original motion, although they have the same basic
6-tuple and hence are close to each other in the vicinity of $p$. We shall
refer to them as a pair of \emph{companion solutions}.

The original shape curve $\gamma^{\ast}$ has, in fact, been "found" by
numerical experiments to be periodic, see
http://three-body.ipb.ac.rs/henon.php and the example Henon 2. But we leave
open the question of whether its companion is also periodic or if the two
shape curves are globally close to each other.
\end{remark}

\end{document}